\documentclass{article}

\usepackage[latin1]{inputenc} 
\usepackage[english]{babel}  

\usepackage{amsmath}
\usepackage{amsthm}
\usepackage{amssymb}
\usepackage[all]{xy}

\usepackage{enumerate}
\usepackage{latexsym} 
\usepackage{fancyhdr,url,array,calc,float,graphicx,epsfig} 
\usepackage{epic}
\usepackage{eepic}

\newcommand{\ignore}[1]{}

\def\squareforqed{\hbox{\rlap{$\sqcap$}$\sqcup$}}
\def\qed{\ifmmode\squareforqed\else{\unskip\nobreak\hfil
\penalty50\hskip1em\null\nobreak\hfil\squareforqed
\parfillskip=0pt\finalhyphendemerits=0\endgraf}\fi}

\newtheorem{theorem}{Theorem}[section]
\newtheorem{lemma}[theorem]{Lemma}
\newtheorem{corollary}[theorem]{Corollary}
\newtheorem{definition}[theorem]{Definition}

\usepackage{ifthen}
\renewenvironment{proof}[1][Default]
{\ifthenelse{\equal{Default}{#1}}{\noindent\textit{Proof. }}{\noindent\textit{Proof (#1). }}}
{}

\newcommand{\flb}[1]{\noindent \textbf{#1}}
\newcommand{\prob}[1]{{\sc #1}}
\newcommand{\alg}[1]{\mathcal{#1}} 

\newcommand{\cc}[1]{\textnormal{\textbf{#1}}} 
\newcommand{\etal}{~et~al.}
\newcommand{\lattleq}[0]{\sqsubseteq}

\newcommand{\lattl}[0]{\sqsubset}

\newcommand{\lub}[0]{\sqcup}
\newcommand{\glb}[0]{\sqcap}
\usepackage{stmaryrd} 
\newcommand{\biglub}[0]{\bigsqcup}

\newcommand{\maxarg}[0]{\mathop{\mbox{\rm max arg}}}

\def\tup#1{\mathchoice{\mbox{\boldmath$\displaystyle#1$}}
{\mbox{\boldmath$\textstyle#1$}}
{\mbox{\boldmath$\scriptstyle#1$}}
{\mbox{\boldmath$\scriptscriptstyle#1$}}}

\newcommand{\arr}[1]{\ar@{-}[#1] |-{\object@{>}} }
\newcommand{\uarr}[1]{\ar@{-}[#1]}

\title{On the Complexity of Submodular Function Minimisation on Diamonds}
\author{Fredrik Kuivinen\footnote{Department of Computer and Information Science,
Link\"{o}pings Universitet, SE-581 83, Link\"{o}ping, Sweden. E-mail:
freku@ida.liu.se}}

\begin{document}
\maketitle

\begin{abstract}
Let $(L; \glb, \lub)$ be a finite lattice and let $n$ be a positive
integer. A function $f : L^n \rightarrow \mathbb{R}$ is said to be
\emph{submodular} if $f(\tup{a} \glb \tup{b}) + f(\tup{a} \lub
\tup{b}) \leq f(\tup{a}) + f(\tup{b})$ for all $\tup{a}, \tup{b} \in
L^n$. In this paper we study submodular functions when $L$ is a
\emph{diamond}. Given oracle access to $f$ we are interested in
finding $\tup{x} \in L^n$ such that $f(\tup{x}) = \min_{\tup{y} \in
L^n} f(\tup{y})$ as efficiently as possible.
We establish
\begin{itemize}
  \item a min--max theorem, which states that the minimum of the
  submodular function is equal to the maximum of a certain function
  defined over a certain polyhedron; and

  \item a good characterisation of the minimisation problem, i.e., we
    show that given an oracle for computing a submodular $f : L^n
    \rightarrow \mathbb{Z}$ and an integer $m$ such that
    $\min_{\tup{x} \in L^n} f(\tup{x}) = m$, there is a proof of this
    fact which can be verified in time polynomial in $n$ and
    $\max_{\tup{t} \in L^n} \log |f(\tup{t})|$; and

  \item a pseudo-polynomial time algorithm for the minimisation
  problem, i.e., given an oracle for computing a submodular $f : L^n
  \rightarrow \mathbb{Z}$ one can find $\min_{\tup{t} \in L^n}
  f(\tup{t})$ in time bounded by a polynomial in $n$ and
  $\max_{\tup{t} \in L^n} |f(\tup{t})|$.
\end{itemize}
\ignore{We also show that for any modular lattice $L$ and sublattice $S$ of
$L$ that if it is possible to minimise submodular functions over $L^n$
in time polynomial in $n$, then it is possible to minimise submodular
functions over $S^n$ in time polynomial in $n$.}
\ignore{This work partially answers open questions posed by (Cohen,
Cooper, Jeavons and Krokhin, 2005)\ignore{\cite{supmod-maxcsp}} and
(Iwata, Fleischer and Fujishige,
2001)\ignore{\cite{submod-min-P-iwata}}.}
\end{abstract}

\section{Introduction}
Let $V$ be a finite set and let $f$ be a function from $2^V$ to
$\mathbb{R}$. The function $f$ is said to be \emph{submodular} if
$f(A \cup B) + f(A \cap B) \leq f(A) + f(B)$
for all $A, B \subseteq V$. In the sequel we will call such functions
\emph{submodular set functions}. Submodular set functions shows up in
various fields including combinatorial optimisation, graph
theory~\cite{submod-graph-theory}, game theory~\cite{cores-convex-games}, information
theory~\cite{capacity-region-ma} and statistical
physics~\cite{optimal-coop-potts}. Examples include the cut function
of graphs and the rank function of matroids. There is also a
connection between submodular function minimisation and convex
optimisation. In particular, submodularity can be seen as a discrete
analog of convexity~\cite{submod-prog,submod-conv}. We refer the
reader to~\cite{sfm-book,sfm-survey,sfm-survey-mccormick} for a
general background on submodular set functions.

Given a submodular set function $f : 2^V \rightarrow \mathbb{R}$ there
are several algorithms for finding minimisers of $f$, i.e., finding a
subset $X \subseteq V$ such that $f(X) = \min_{Y \subseteq V} f(Y)$,
in time polynomial in $|V|$. The first algorithm for finding such
minimisers in polynomial time is due to
Gr\"otschel\etal~\cite{ellips-combopt}. However, this algorithm is
based on the Ellipsoid algorithm and hence its usefulness in practise
is limited. Almost two decades later two combinatorial algorithms were
found independently by Schrijver~\cite{submod-min-P-alex} and
Iwata\etal~\cite{submod-min-P-iwata}. More recently the running times
have been improved. The currently fastest strongly polynomial time algorithm
is due to Orlin~\cite{james-faster-submod} and the fastest weakly polynomial
time algorithm is due to Iwata~\cite{faster-submod}. In these algorithms the
submodular set function is given by a value-giving oracle for $f$
(i.e., presented with a subset $X \subseteq V$ the oracle computes
$f(X)$).

In this paper we investigate a more general notion of
submodularity. Recall that a \emph{lattice} is a partially ordered set
in which each pair of elements have a least upper bound (join, $\lub$)
and a greatest lower bound (meet, $\glb$). Given a finite lattice $\alg{L}$
(all lattices in this paper are finite) and a positive integer $n$ we
can construct the \emph{product lattice} $\alg{L}^n$. Meet and join for
$\alg{L}^n$ are then defined coordinate-wise by meet and join in $\alg{L}$. We say
that a function $h : \alg{L}^n \rightarrow \mathbb{R}$ is submodular if
$h(\tup{a} \glb \tup{b}) + h(\tup{a} \lub \tup{b}) \leq h(\tup{a}) +
h(\tup{b})$ for all $\tup{a}, \tup{b} \in \alg{L}^n$. Note that the subsets
of $V$ can be seen as a lattice with union as join and intersection as
meet (this lattice is a product of the two element lattice). Hence,
this notion of submodularity is a generalisation of submodular set
functions. For a fixed finite lattice $\alg{L}$ we are interested in the
submodular function minimisation (SFM) problem:

\medskip
\noindent \textsc{Instance:} An integer $n \geq 1$ and a submodular function $f$ on $\alg{L}^n$.

\noindent \textsc{Goal:} Find $\tup{x} \in \alg{L}^n$ such that $f(\tup{x}) = \min_{\tup{y} \in \alg{L}^n} f(\tup{y})$.
\medskip

Following~\cite{sfm-prod} we denote this problem by SFM$(\alg{L})$. SFM$(\alg{L})$
is said to be \emph{oracle-tractable} if the problem can be solved in
time polynomial in $n$ (provided that we have access to a value-giving
oracle for $f$ and that we can assume that $f$ is submodular, i.e., it is
a promise problem). This definition
naturally leads to the following question: is SFM$(\alg{L})$
oracle-tractable for all finite lattices $\alg{L}$? (This question was, as
far as we know, first asked by Cohen\etal~\cite{supmod-maxcsp}.)

Schrijver~\cite{submod-min-P-alex} showed that given a sublattice $S$
of $2^V$ (i.e., $S \subseteq 2^V$ and for any $X, Y \in S$ we have $X
\cap Y, X \cup Y \in S$) and submodular function $f : S \rightarrow
\mathbb{R}$ a minimiser of $f$ can be found in time polynomial in
$n$. In particular, this implies that for any distributive lattice $\alg{L}$
the problem SFM$(\alg{L})$ is oracle-tractable.
Krokhin and Larose~\cite{sfm-prod} showed that certain constructions
on lattices preserve oracle-tractability of SFM. In particular, they
showed that if $X$ is a class of lattices such that SFM$(\alg{L})$ is
oracle-tractable for every $\alg{L} \in X$, then so is SFM$(\alg{L}')$ where $\alg{L}'$
is a \emph{homomorphic image} of some lattice in $X$, a \emph{direct
product} of some lattices in $X$, or contained in the \emph{Mal'tsev
product} $X \circ X$. We will not define these constructions here and
refer the reader to~\cite{sfm-prod} instead.

A lattice $\alg{L}$ is a \emph{diamond} if the elements of the lattice form
a disjoint union of $\{0_{\alg{L}}, 1_{\alg{L}}\}$ and $A$, for some finite set $A$
such that $|A| \geq 3$. Here $0_{\alg{L}}$ is the bottom element of $\alg{L}$, and
$1_{\alg{L}}$ is the top element of $\alg{L}$, and all elements in $A$ (called the \emph{atoms}) are
incomparable to each other. See Figure~\ref{diam:fig:diam} for a diagram of the five element diamond.
We want to emphasise that diamonds have a
different structure compared to the lattices defined by union and
intersection. In particular, diamonds are not \emph{distributive},
that is they do \emph{not} satisfy $x \glb (y \lub z) = (x \glb y)
\lub (x \glb z)$ for all $x, y, z \in \alg{L}$. We will denote the diamond
with $k$ atoms by $\alg{M}_k$. In Sections~\ref{diam:sec:minmax}, \ref{diam:sec:good}
and~\ref{diam:sec:find-min} the complexity of SFM$(\alg{M}_k)$ is
investigated. In the approach taken in this paper the difficult case
is $k = 3$---the proofs for the $k=3$ case generalises
straightforwardly to an arbitrary $k$. We note that none of the
diamonds are captured by the combination of the results found
in~\cite{sfm-prod,submod-min-P-alex} (a proof of this fact can be
found in~\cite{sfm-prod}).

\begin{figure}
\[
    \xymatrix{             & 1_\alg{M}        &                       \\
                      \bullet \uarr{ur} & \bullet \uarr{u} & \bullet \uarr{ul}           \\
                                  & 0_\alg{M} \uarr{u} \uarr{ul} \uarr{ur} & }
\]
\caption{The five element diamond.} \label{diam:fig:diam}
\end{figure}
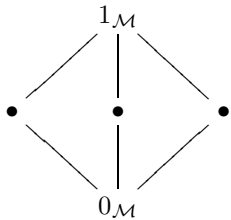

\flb{Results and techniques.}
The first main result in this paper is a min--max theorem for
SFM$(\alg{M}_k)$ which is stated as Theorem~\ref{diam:th:minmax}. This result looks quite similar to Edmonds'
min--max theorem for submodular set functions~\cite{sfm-minmax} (we
present Edmonds' result in Section~\ref{diam:sec:submod-set-func}).  The
key step in the proof of this result is the definition of a certain
polyhedron, which depends on $f$.

The second main result is a \emph{good characterisation} of SFM$(\alg{M}_k)$
(Theorem~\ref{diam:th:coNP}). That is, we prove that given a submodular $f
: \alg{M}_k^n \rightarrow \mathbb{Z}$ and integer $m$ such that
$\min_{\tup{x} \in \alg{L}^n} f(\tup{x}) = m$, there is a proof of this fact
which can be verified in time polynomial in $n$ and $\max_{\tup{y} \in
\alg{L}^n} \log |f(\tup{y})|$ (under the assumption that $f$ is
submodular). \ignore{FIXME: bort? eller formulera om? eller ha kvar?}
This can be seen as placing SFM$(\alg{M}_k)$ in the appropriately modified
variant of \cc{NP} $\cap$ \cc{coNP} (the differences from our setting
to an ordinary optimisation problem is that we are given oracle access
to the function to be minimised and we assume that the given function
is submodular).  The proof of this result makes use of Carath{\'e}odory's
theorem and of the known polynomial-time algorithms for minimising
submodular set functions. We also need our min--max theorem.

The third result is a pseudo-polynomial time algorithm for SFM$(\alg{M}_k)$
(see Section~\ref{diam:sec:find-min}). We show that SFM$(\alg{M}_k)$ can be
solved in time polynomial in $n$ and $\max_{\tup{t} \in \alg{M}_k^n}
|f(\tup{t})|$. The main part of the algorithm consists of a nested
application of the Ellipsoid algorithm. We also need to prove that the
polyhedrons we associate with submodular functions are
$1/2$-integral. An interesting and challenging open problem is to construct
an algorithm with running time
polynomial in $n$ and $\max_{\tup{t} \in \alg{M}_k^n} \log |f(\tup{t})|$.

\ignore{\textbf{FIXME: Describe this result more}}

\ignore{A \emph{sublattice} $S$ of a lattice $L$ is a set $S \subseteq L$ such
that $x, y \in S \Rightarrow x \glb y, x \lub y \in S$. Recall that a
lattice is \emph{modular} if and only if for all $u, v, w \in L$ such
that $u \lattleq w$, we have $u \lub (v \glb w) = (u \lub v) \glb
w$. In particular, the diamonds and any direct product of diamonds are
modular. In Section~\ref{diam:sec:sublatt} we show that if $L$ is a modular
lattice such that SFM$(L)$ is oracle-tractable and $S$ is a sublattice
of $L$, then SFM$(S)$ also is oracle-tractable. This is proved by
modifying a construction due to Schrijver~\cite{submod-min-P-alex}.}

Our results applies to diamonds, however, as mentioned above,
in~\cite{sfm-prod} two constructions on lattices (Mal'tsev products
and homomorphic images) are shown to preserve tractability results for
SFM. By combining these constructions with the results in this paper
one gets tractability results for a much larger class of lattices than
just diamonds.~\footnote{In~\cite{sfm-prod} these constructions are
shown to preserve oracle-tractability and not solvability in
pseudo-polynomial time. However, it is straightforward to adapt the
proofs to the pseudo-polynomial case.} In particular, by the results
in this paper there is a pseudo-polynomial time algorithm for
minimising submodular functions over products of the lattice in
Figure~\ref{diam:fig:diam2}.

\begin{figure}
\centering
\includegraphics[width=3cm]{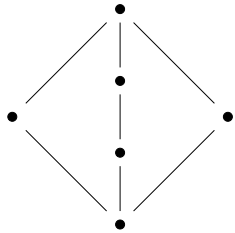}
\caption{A lattice which can be shown to admit a pseudo-polynomial
time algorithm for the submodular function minimisation problem. This
lattice is a Mal'tsev product of a diamond and the two element
lattice. By the results in this paper and the constructions in~\cite{sfm-prod}
this lattice gives a new tractable constraint language for \prob{Max
CSP}.}
\label{diam:fig:diam2}
\end{figure}

\flb{Connections to other problems.}
Minimising submodular functions on certain modular non-distributive
(the diamonds are modular and non-distributive) lattices has
applications to canonical forms of partitioned
matrices~\cite{block-triang-part-matrix,minimax-dm-decomp}. Finding a
polynomial time algorithm for minimising submodular functions on these
lattices was mentioned as an open problem
in~\cite{submod-min-P-iwata}.

The notion of submodular functions over arbitrary finite lattices
plays an important role in the complexity of the maximum
constraint satisfaction problem (\prob{Max CSP}). This connection was
first observed in~\cite{supmod-maxcsp} and in later papers the
connection was explored
further~\cite{maxcsp-fixed-journal,maxcsp-three}.
The connection between submodular function minimisation and \prob{Max
CSP} is that by proving oracle-tractability for new lattices for the
SFM problem implies tractability results (solvability in polynomial
time) for certain restrictions (so called \emph{constraint language
restrictions}) of \prob{Max CSP}.
By constructing algorithms for SFM with running times bounded by a
polynomial in $n$ and $\max_{\tup{t} \in \alg{M}_k^n} |f(\tup{t})|$, as we
do in Section~\ref{diam:sec:find-min}, one gets solvability in polynomial
time for the unweighted variant of \prob{Max CSP} (with appropriate
restrictions).
Providing good characterisations of SFM$(\alg{L})$, as we do in
Section~\ref{diam:sec:good}, implies \cc{coNP} containment results for
\prob{Max CSP} (with appropriate restriction). As \prob{Max CSP} is trivially in \cc{NP} we get
containment in \cc{NP} $\cap$ \cc{coNP} for these restrictions.  We
refer the reader to~\cite{supmod-maxcsp,sfm-prod} for further details
regarding the connection between SFM and \prob{Max CSP}.

In~\cite{sfm-prod} it is shown that the restrictions of \prob{Max CSP}
which one gets from the diamonds can be solved in polynomial
time. This means that the results for the diamonds in this paper does
not directly imply new tractability results for \prob{Max
CSP}. However, as mentioned in the previous section one can combine
the results in this paper with the lattice constructions
in~\cite{sfm-prod} to get tractability results for a larger class of
lattices which implies tractability results for new constraint
language restrictions of \prob{Max CSP}. (We again refer to
Figure~\ref{diam:fig:diam2} for an example of such a lattice.)

There is also a connection between SFM over lattices to the
Valued Constraint Satisfaction Problem
(VCSP). See~\cite{Cohen:etal:ai2006} for more information on VCSP. The
connection is very similar to the connection to \prob{Max CSP},
proving tractability results for new lattices for SFM implies new
tractable restrictions of VCSP. For VCSP there was, before the results
in this paper, no known non-trivial algorithms for the restrictions one
obtains from the diamonds.

We note that Raghavendra~\cite{optallcsp} recently proved almost
optimal results for the \emph{approximability} of \prob{Max CSP} for
constraint language restrictions, assuming that the unique games
conjecture (UGC) holds. However, for the cases which are solvable to
optimality the results in~\cite{optallcsp} gives us polynomial-time
approximation schemes. This should be compared to the connection
described above which gives polynomial time algorithms for some of
these cases.

\flb{Organisation.}  This paper is organised as follows, in
Section~\ref{diam:sec:submod-set-func} we give a short background on
submodular set functions, in Section~\ref{diam:sec:prel} we introduce the
notation we use, in Section~\ref{diam:sec:minmax} we prove our first main
result---the min--max theorem for submodular functions over
diamonds. The good characterisation is given in
Section~\ref{diam:sec:good}. \ignore{In Section~\ref{diam:sec:ellips} we give
some definitions and results which are related to the Ellipsoid
algorithm.} In Section~\ref{diam:sec:find-min} where we give the
pseudo-polynomial time algorithm for the minimisation problem.
\ignore{In Section~\ref{diam:sec:sublatt} we prove that sublattices of
oracle-tractable modular lattices are also oracle-tractable.} Finally,
in Section~\ref{diam:sec:concl} we give some conclusions and open
problems. 

\section{Background on Submodular Set Functions} \label{diam:sec:submod-set-func}
In this section we will give a short background on Edmonds' min--max
theorem for submodular set functions. This result was first proved by
Edmonds in~\cite{sfm-minmax}, but see also the
surveys~\cite{sfm-survey,sfm-survey-mccormick}. Let $V$ be a finite set. For a vector
$\tup{x} \in \mathbb{R}^V$ (i.e., $\tup{x}$ is a function from $V$
into $\mathbb{R}$) and a subset $Y \subseteq V$ define $\tup{x}(Y) =
\sum_{y \in Y} \tup{x}(y)$.  We write $\tup{x} \leq 0$ if $\tup{x}(v)
\leq 0$ for all $v \in V$ and $\tup{x}^-$ for the vector in which
coordinate $v$ has the value $\min \{0, \tup{x}(v) \}$. Let $f$ be a
submodular set function $f : 2^V \rightarrow \mathbb{R}$ such that
$f(\emptyset) = 0$ (this is not really a restriction, given a
submodular function $g$ we can define a new function $g'(X) = g(X) -
g(\emptyset)$, $g'$ satisfies $g'(\emptyset) = 0$ and is
submodular). The \emph{submodular polyhedron} and the \emph{base
polyhedron} defined by
\begin{align}
&P(f) = \{ \tup{x} \in \mathbb{R}^V \mid \forall Y \subseteq V, \tup{x}(Y) \leq f(Y) \}, \text{ and}  \notag \\
&B(f) = \{ \tup{x} \in \mathbb{R}^V \mid \tup{x} \in P(f), \tup{x}(V) = f(V) \} \notag
\end{align}
often play an important role in results related to submodular set
functions.  Edmonds~\cite{sfm-minmax} proved the following min--max
theorem
\begin{align}
\min_{X \subseteq V} f(X) &= \max \{ \tup{x}(V) \mid \tup{x} \in P(f), \tup{x} \leq 0 \} \notag \\
                          &= \max \{ \tup{x}^-(V) \mid \tup{x} \in B(f) \} \label{diam:eq:set-minmax} .
\end{align}
In Section~\ref{diam:sec:minmax} we give an analog
to~\eqref{diam:eq:set-minmax} for submodular functions over
diamonds.

\section{Preliminaries} \label{diam:sec:prel}
For a positive integer $n$, $[n]$ is the set $\{1, 2, \ldots, n\}$.
Given a lattice $(L, \glb, \lub)$ and $x, y \in L$ we write $x
\lattleq y$ if and only if $x \glb y = x$ (and hence $x \lub y =
y$). We write $x \lattl y$ if $x \lattleq y$ and $x \neq y$.  As
mentioned in the introduction, given a positive integer $n$, we can
construct the product lattice $L^n$ from $L$.  The top and bottom
elements of $L^n$ are denoted by $\tup{1}_{L^n}$ and $\tup{0}_{L^n}$,
respectively. We write $x \prec y$ if $x$ is covered by $y$ (that is,
if $x \lattl y$, and there is no $z \in L$ such that $x \lattl z
\lattl y$).

Recall that the diamonds are modular lattices (the rank function
$\rho$ is defined by $\rho(0_\alg{M}) = 0$, $\rho(a) = 1$ for all $a \in A$
and $\rho(1_\alg{M}) = 2$). As direct products of modular lattices also are
modular lattices it follows that direct products of diamonds are
modular lattices.

For a set $X$ we let $\mathbb{R}^{[n] \times X}$ be the set of
functions mapping $[n] \times X$ into $\mathbb{R}$. Such functions
will be called \emph{vectors} and can be seen as vectors indexed by
pairs from $[n] \times X$. For $\tup{x}, \tup{y} \in \mathbb{R}^{[n]
\times X}$ and $\alpha \in \mathbb{R}$ we define $\alpha \tup{x},
\tup{x} + \tup{y}, \tup{x}^- \in \mathbb{R}^{[n] \times X}$ as
$(\alpha \tup{x})(i, x) = \alpha \tup{x}(i, x)$, $(\tup{x} +
\tup{y})(i, x) = \tup{x}(i, x) + \tup{y}(i, x)$, and $\tup{x}^-(i, x)
= \min \{0, \tup{x}(i, x) \}$ for all $i \in [n]$ and $x \in X$,
respectively. If $\tup{x}(i, x) \leq 0$ for all $i \in [n]$ and $x \in
X$ we write $\tup{x} \leq 0$. For $i \in [n]$ we use $\tup{x}(i)$ to
denote the function $x' \in \mathbb{R}^X$ such that $\tup{x}(i,
x) = x'(x)$ for all $x \in X$.

For $i \in [n]$ and $a \in A$ let $\tup{\chi_{i, a}} \in
\mathbb{R}^{[n] \times A}$ be the vector such that $\tup{\chi_{i,
a}}(i, a) = 1$ and $\tup{\chi_{i, a}}(i', a') = 0$ for $(i', a') \neq
(i, a)$.  (So $\tup{\chi_{i,a}}$ is the unit vector for the coordinate
$(i,a)$.) Similarly, we use $\tup{\chi_i}$ to denote the vector
$\sum_{a \in A} \tup{\chi}_{i, a}$.  For a vector $\tup{x} \in
\mathbb{R}^{[n] \times A}$ and tuple $\tup{y} \in \alg{M}^n$ we define
\[
\tup{x}(\tup{y}) = \sum_{i = 1}^n g(\tup{x}(i), \tup{y}(i))
\]
where the function $g : \mathbb{R}^A \times \alg{M} \rightarrow \mathbb{R}$ is
defined by
\[
g(x, y) =
\left\{
\begin{array}{ll}
0                & \textrm{if } y = 0_\alg{M}, \\
x(y)             & \textrm{if } y \in A, \textrm{ and} \\
\max_{a, a' \in A, a \neq a'} x(a) + x(a') & \textrm{otherwise (if $y = 1_\alg{M}$).}
\end{array}
\right.
\]
(This should be compared to how applying a vector to a subset is
defined for submodular set functions, see~\cite{sfm-minmax}.)  For $\tup{x}, \tup{x'}
\in \mathbb{R}^{[n] \times A}$ we denote the usual scalar product by
$\langle \tup{x}, \tup{x}' \rangle$, so
\[
\langle \tup{x}, \tup{x}'\rangle = \sum_{i=1}^n \sum_{x \in A} \tup{x}(i, x) \tup{x}'(i, x) .
\]
Let $f$ be a submodular function on $\alg{M}^n$ such that $f(\tup{0}_{M^n})
\geq 0$. We define $P_M(f)$ and $B_M(f)$ as follows,
\begin{align}
&P_M(f) = \left\{ \tup{x} \in \mathbb{R}^{[n] \times A} \ \Big| \ \forall \tup{y} \in \alg{M}^n, \tup{x}(\tup{y}) \leq f(\tup{y}) \right\}, \text{ and} \notag \\
&B_M(f) = \left\{ \tup{x} \in \mathbb{R}^{[n] \times A} \ \Big| \ \tup{x} \in P_M(f), \tup{x}(\tup{1}_{\alg{M}^n}) = f(\tup{1}_{\alg{M}^n}) \right\} . \notag
\end{align}
Due to the definition of $g$ it is not hard to see that $P_M(f)$ is a
polyhedron. Note that if $\tup{t}$ contains at least one $1_\alg{M}$, then
$\tup{t}$ induce more than one linear inequality. If $\tup{t}$
contains no $1_\alg{M}$, then $\tup{t}$ only induce one linear
inequality. In general, a tuple with $m$ occurrences of $1_\alg{M}$ induces
${|A| \choose 2}^m$ linear inequalities. We use $I(\tup{t})$ to denote
the set of all vectors $\tup{e} \in \mathbb{R}^{[n] \times A}$ such
that $\tup{e}$ represents an inequality induced by $\tup{t}$ (that is,
an inequality of the form $\langle \tup{e}, \tup{x} \rangle \leq
f(\tup{t})$, where $\tup{e} \in I(\tup{t})$). Given a vector $\tup{x}
\in P_M(f)$ we say that a tuple $\tup{t} \in \alg{M}^n$ such that
$\tup{x}(\tup{t}) = f(\tup{t})$ is \emph{$\tup{x}$-tight}. 

We will also need the following definition.
\begin{definition}[Unified Vector for Diamonds] \label{diam:def:unif}
A vector $\tup{x} \in \mathbb{R}^A$ is \emph{unified} if there is an
atom $p \in A$ such that
\begin{itemize}
\item if $x, y \in A, x, y \neq p$, then $\tup{x}(x) = \tup{x}(y)$; and
\item if $a \in A$, then $\tup{x}(p) \geq \tup{x}(a)$.
\end{itemize}
We extend the definition of unified vectors to the vectors in
$\mathbb{R}^{[n] \times A}$ by saying that $\tup{x} \in
\mathbb{R}^{[n] \times A}$ is unified if $x \mapsto \tup{x}(i, x)$ is
unified for each $i \in [n]$.
\end{definition}

If the submodular inequality is strict for all incomparable pair of
elements then we say that the function is \emph{strictly submodular}.

\section{A Min--Max Theorem} \label{diam:sec:minmax}
The main results in this section are Theorem~\ref{diam:th:minmax} and
Theorem~\ref{diam:th:Pf=Bf}. We start by a lemma which shows that $B_M(f)$
is non-empty for any submodular function which maps the bottom of the
lattice to a non-negative value.

\begin{lemma} \label{diam:lem:greedy-Bf}
  Let $f : \alg{M}^n \rightarrow \mathbb{R}$ be submodular such that
  $f(\tup{0}) \geq 0$. There is a vector $\tup{x} \in \mathbb{R}^{[n]
  \times A}$ such that
  \begin{itemize}
    \item $\tup{x}$ is unified; and

    \item $\tup{x}(\tup{v_i}) = f(\tup{v_i})$ for all $i \in [n]$; and

    \item $\tup{x}(\tup{v_i}[i+1 = p_{i+1}]) = f(\tup{v_i}[i+1 =
      p_{i+1}])$ for all $i \in \{0, 1, \ldots, n-1\}$, where for $i
      \in [n]$, $p_{i}$ is the atom in Definition~\ref{diam:def:unif} for
      the vector $x \mapsto \tup{x}(i, x)$.
  \end{itemize}
  Furthermore, if $f$ is integer-valued, then $\tup{x}$ can be chosen
  to be integer-valued.
\end{lemma}
\begin{proof}
  Given a submodular $f : \alg{M}^n \rightarrow \mathbb{R}$ we will
  construct a vector $\tup{x}$ which satisfies the requirements in the
  lemma. To do this we define a sequence of atoms $p_i$ for $i \in
  [n]$ inductively. To start the inductive definition let $p_1 \in
  \maxarg_{a \in A} f(\tup{v_0}[1 = a])$ and set $\tup{x}(1, p_1) =
  f(\tup{v_0}[1 = p_1])$. For the general case, choose $p_i \in A$ so
  that
  \[
  p_i \in \maxarg_{a \in A} f(\tup{v_i}[i+1 = a]) .
  \]
  For $i \in [n]$ set
  \begin{align}
    \tup{x}(i+1, p_{i+1}) = f(\tup{v_i}[i+1 = p_{i+1}]) - f(\tup{v_i}), \label{diam:eq:pi+1}
  \end{align}
  and for $a \in A, a \neq p_{i+1}$ set
  \begin{align}
    \tup{x}(i+1, a) = f(\tup{v_{i+1}}) - f(\tup{v_i}[i+1 = p_{i+1}]) . \label{diam:eq:a}
  \end{align}

  \flb{Claim A. If $a \in A$, then $\tup{x}(i+1, p_{i+1}) \geq
  \tup{x}(i+1, a)$.}

  Assume, without loss of generality, that $a \neq p_{i+1}$. We now
  get
  \begin{align}
  &f(\tup{v_{i+1}}) + f(\tup{v_i}) &\leq \notag \\
  &f(\tup{v_i}[i+1 = p_{i+1}]) + f(\tup{v_i}[i+1 = a]) &\leq \notag \\
  &2 f(\tup{v_i}[i+1 = p_{i+1}]) \notag
  \end{align}
  where the first inequality holds due to the submodularity of $f$ and
  the second inequality follows from our choice of $p_{i+1}$. This is
  equivalent to
  \[
  f(\tup{v_{i+1}}) - f(\tup{v_i}[i+1 = p_{i+1}]) \leq f(\tup{v_i}[i+1 = p_{i+1}]) - f(\tup{v_i})
  \]
  which is what we wanted to prove.  \qed

  For $i \in [n]$ and $j \in \{1, 2\}$ we define $c_{i,j}$ as
  $c_{i,1} = p_i$ and $c_{i,2} = 1_\alg{M}$. \qed

  \flb{Claim B. $\tup{x}(\tup{v_i}[i+1 = c_{i,j}]) = f(\tup{v_i}[i+1 =
  c_{i,j}])$ for all $(i,j) \in \{0,1,\ldots,n-1\} \times \{1,2\}$.}

  We prove this by induction over the pairs $(i,j)$ ordered
  lexicographically (so $(i, j) \leq (i', j')$ if and only if $i < i'$ or ($i =
  i'$ and $j \leq j'$)). With the pair $(i,j)$ we associate the tuple
  $\tup{v_i}[i+1 = c_{i,j}]$. Note that $(i,j) \leq (i',j')$ if and
  only if $\tup{v_i}[i+1 = c_{i+1,j}] \lattleq \tup{v_{i'}}[i'+1 =
  c_{i'+1,j'}]$. As $p_1 \in \maxarg_{a \in A} f(\tup{v_0}[1 = a])$
  the claim clearly holds for $(i,j) = (0,1)$. Now assume that it
  holds for all pairs $(i',j')$ such that $(i', j') \leq (i,j)$. If $j
  = 1$ then the next pair is $(i, 2)$ and we get
  \begin{align}
    \tup{x}(\tup{v_i}[i+1 = c_{i,2}]) &= \tup{x}(\tup{v_i}[i+1 = p_{i+1}]) + \tup{x}(i+1, a) \notag \\
                                        &= f(\tup{v_i}[i+1 = p_{i+1}]) + f(\tup{v_{i+1}}) - f(\tup{v_i}[i+1 = p_{i+1}]) \notag \\
                                        &= f(\tup{v_{i+1}}) . \notag
  \end{align}
  Here the first inequality follows from the definition of
  $\tup{x}(\cdot)$ and Claim~A. The second equality follows from the
  induction hypothesis and~\eqref{diam:eq:a}. If $j = 2$ the next pair is
  $(i+1,1)$ and we get
  \begin{align}
    \tup{x}(\tup{v_{i+1}}[i+2 = c_{i+2,1}]) &= \tup{x}(\tup{v_{i+1}}) + \tup{x}(i+2, p_{i+2}) \notag \\
                                            &= f(\tup{v_{i+1}}) + f(\tup{v_{i+1}}[i+2 = p_{i+2}]) - f(\tup{v_{i+1}}) \notag \\
                                            &= f(\tup{v_{i+1}}[i+2 = p_{i+2}]) \notag
  \end{align}
  As above the first equality follows from the definition of
  $\tup{x}(\cdot)$ and Claim~A. The second equality follows from the
  induction hypothesis and~\eqref{diam:eq:a}. \qed

  By Claim~A it follows that $\tup{x}$ is unified. By Claim~B
  $\tup{x}$ satisfies the second condition in the statement of the
  lemma. It is easy to see that if $f$ is integer-valued, then so is
  $\tup{x}$. \qed
\end{proof}

\begin{lemma} \label{diam:lem:spmod->inPf}
  Let $f : \alg{M}^n \rightarrow \mathbb{R}$ be submodular such that
  $f(\tup{0}) \geq 0$. Let $\tup{x}$ be a vector in $\mathbb{R}^{[n]
  \times A}$. If for each $i \in [n]$ there is an atom $p_i$ such that
  \begin{itemize}
    \item for all $i \in [n]$ we have $\tup{x}(\tup{v_i}) = f(\tup{v_i})$, and
    \item for all $i \in \{0,1,\ldots,n-1\}$ we have
      $\tup{x}(\tup{v_i}[i+1 = p_{i+1}]) = f(\tup{v_i}[i+1 =
      p_{i+1}])$,
  \end{itemize}
  then $\tup{x} \in P_M(f)$.
\end{lemma}
\begin{proof}
For $i \in [n]$ and $j \in \{1,2\}$ we define $c_{i,j}$ as follows
$c_{i,1} = p_i$ and $c_{i,2} = 1_\alg{M}$. We will prove by induction that
$\tup{x}(\tup{y}) \leq f(\tup{y})$ for all $\tup{y} \in \alg{M}^n$. As in
the proof of Claim~B in Lemma~\ref{diam:lem:greedy-Bf} the induction will
be over the pairs $\{0,1,\ldots,n-1\} \times \{1,2\}$ ordered
lexicographically. With the pair $(i,j)$ we associate the tuples
$\tup{y}$ such that $\tup{y} \lattleq \tup{v_i}[i+1 = p_{i,j}]$.

As
\[
\tup{x}(\tup{v_0}) = \tup{x}(\tup{0_{\alg{M}^n}}) = 0 \text{ and } f(\tup{0_{\alg{M}^n}}) \geq 0
\]
and
\[
\tup{x}(\tup{v_0}[1 = p_1]) = f(\tup{v_0}[1 = p_1])
\]
the statement holds for the pair $(0,1)$ (which corresponds to
$\tup{y} \lattleq \tup{0_{\alg{M}^n}}[1 = p_1]$). Let $i \in
\{0,1,\ldots,n-1\}$, $j \in \{1,2\}$, and $\tup{y} \in \alg{M}^n, \tup{y}
\lattleq \tup{v_i}[i+1 = c_{i+1,j}]$ and assume that the inequality
holds for all $\tup{y'} \in \alg{M}^n$ such that $\tup{y'} \lattleq
\tup{v_{i'}}[i'+1 = c_{i'+1,j'}]$ where $(i', j')$ is the predecessor
to the pair $(i,j)$. We will prove that the inequality holds for all
$\tup{y} \lattleq \tup{v_i}[i+1 = c_{i+1,j}]$.

To simplify the notation a bit we let $k = i+1$ and $y =
\tup{y}(k)$. If $y = 0_\alg{M}$ we are already done, so assume that $y \neq
0_\alg{M}$. If $y = p_{k}$ let $c = 0_\alg{M}$, if $y \in A, y \neq p_{k}$ let
$c = p_{k}$ and otherwise, if $y = 1_\alg{M}$ let $c = p_{k}$. Now,
\begin{align}
\tup{x}(\tup{y}) &\leq \tup{x}(\tup{v_i}[k = y \lub c])  - \tup{x}(\tup{v_i}[k = c]) + \tup{x}(\tup{y}[k = y \glb c]) \notag \\
                 &\leq \tup{x}(\tup{v_i}[k = y \lub c])  - \tup{x}(\tup{v_i}[k = c]) + f(\tup{y}[k = y \glb c]) \notag \\
                 &\leq f(\tup{v_i}[k = y \lub c]) - f(\tup{v_i}[k = c]) + f(\tup{y}[k = y \glb c]) \notag \\
                 &\leq f(\tup{y}) . \notag 
\end{align}
The first inequality follows from the supermodularity of
$\tup{x}$. The second inequality follows from the induction hypothesis
and the fact that $y \glb c \lattl y$ and $y \glb c \in \{0_\alg{M},
p_k\}$. The third inequality follows from $y \lub c, c \in \{0_\alg{M},
p_{k}, 1_\alg{M}\}$ and the assumptions in the statement of the
lemma. Finally, the last inequality follows from the submodularity of
$f$. \qed
\end{proof}

In the proof of Lemma~\ref{diam:lem:greedy-Bf} the vector $\tup{x} \in
\mathbb{R}^{[n] \times A}$ is constructed with a greedy approach---we
order the coordinates of the vector, $[n] \times A$, in a certain way
and then set each component to its maximum value subject to the
constraints given in the definition of $B_M(f)$. The greedy algorithm
\emph{does not} solve the optimisation problem for $P_M(f)$.  As an
example, let $\alg{M}_3 = (\{0_\alg{M}, 1_\alg{M}, a, b, c\}, \glb, \lub)$ be a diamond
and let $f : \alg{M}_3 \rightarrow \mathbb{R}$ be defined as $f(0_\alg{M}) = 0$,
$f(a) = f(b) = f(c) = f(1_\alg{M}) = 1$. The function $f$ is submodular. Now
let $\tup{c} \in \mathbb{R}^{[1] \times A}$ and $\tup{c}(1, a) =
\tup{c}(1, b) = \tup{c}(1, c) = 1$. From the greedy algorithm we will
get a vector $\tup{x} \in \mathbb{R}^{[1] \times A}$ such that
$\tup{x}(1, a) = 1$ and $\tup{x}(1, b) = \tup{x}(1, c) = 0$ (or some
permutation of this vector). However, the solution to $\max \langle
\tup{c}, \tup{y} \rangle, \tup{y} \in P_M(f)$ is $\tup{y}(1, a) =
\tup{y}(1, b) = \tup{y}(1, c) = 1/2$ and $3/2 = \langle \tup{c},
\tup{y} \rangle > \langle \tup{c}, \tup{x} \rangle = 1$. This example
also shows that the vertices of $P_M(f)$ are not necessarily integer
valued.  This should be compared to submodular set functions, where
the corresponding optimisation problem \emph{is} solved by the greedy
algorithm.~\cite{sfm-survey-mccormick}

Given an algorithm which solves the optimisation problem over $P_M(f)$
in time polynomial in $n$ we can use the equivalence of optimisation
and separation given by the Ellipsoid algorithm to solve the
separation problem for $P_M(f)$ in polynomial time. With such an
algorithm we can decide if $\tup{0} \in P_M(f)$ or not and by a binary
search we can find a minimiser of $f$ in polynomial time. So a
polynomial time algorithm for the optimisation problem over $P_M(f)$
would be desirable. (The approach outlined above can be used to
minimise submodular set functions, see~\cite{ellips-combopt} or,
e.g.,~\cite{ellips-book}.) We present a pseudo-polynomial
algorithm for the optimisation problem in Section~\ref{diam:sec:find-min}
which uses this technique.

We are now ready to state the two main theorems of this section.
\begin{theorem} \label{diam:th:minmax}
  Let $f : \alg{M}^n \rightarrow \mathbb{R}$ be a submodular function such
  that $f(\tup{0}_{\alg{M}^n}) = 0$, then
  \[
  \min_{\tup{x} \in \alg{M}^n} f(\tup{x}) =
  \max \left\{ \tup{z}(\tup{1_{\alg{M}^n}}) \ \big| \ 
      \tup{z} \in P_M(f), \tup{z} \leq 0, \text{ $\tup{z}$ is unified } \right\} .
  \]
  More over, if $f$ is integer-valued then there is an integer-valued
  vector $\tup{z}$ which maximises the right hand side.
\end{theorem}
\begin{proof}
  If $\tup{z} \in P_M(f)$ and $\tup{z} \leq 0$ then
  \[
  \tup{z}(\tup{1_{\alg{M}^n}}) \leq \tup{z}(\tup{y}) \leq f(\tup{y})
  \]
  for any $\tup{y} \in \alg{M}^n$. Hence, LHS $\geq$ RHS holds. Consider the
  function $f' : \alg{M}^n \rightarrow \mathbb{R}$ defined by
  \[
  f'(\tup{x}) = \min_{\tup{y} \lattleq \tup{x}} f(\tup{y}) .
  \]
  Then $P_M(f') \subseteq P_M(f)$.

  \flb{Claim A. $f'$ is submodular.}

  Let $\tup{x'}, \tup{y'} \in
  \alg{M}^n$ and let $\tup{x} \lattleq \tup{x'}, \tup{y} \lattleq \tup{y'}$
  be tuples such that $f'(\tup{x'}) = f(\tup{x})$ and $f'(\tup{y'}) =
  f(\tup{y})$. Now,
  \[
  f'(\tup{x'}) + f'(\tup{y'}) =
  f(\tup{x}) + f(\tup{y}) \geq f(\tup{x} \glb \tup{y}) + f(\tup{x} \lub \tup{y}) \geq f'(\tup{x'} \glb \tup{y'}) + f'(\tup{x'} \lub \tup{y'})
  \]
  where the first equality follows from the definition of $f'$,
  $\tup{x}$ and $\tup{y}$, the first inequality follows from the
  submodularity of $f$ and the second inequality from the definition
  of $f'$ and $\tup{x} \glb \tup{y} \lattleq \tup{x'} \glb \tup{y'}$
  and $\tup{x} \lub \tup{y} \lattleq \tup{x'} \lub \tup{y'}$. \qed

  \flb{Claim B. For any $\tup{z} \in P_M(f')$ we have $\tup{z} \leq
  0$.}

  As $f(\tup{0}_{\alg{M}^n}) = 0$ we have $f'(\tup{x}) \leq 0$ for any
  $\tup{x} \in \alg{M}^n$. For $i \in [n]$ and $a \in A$ define
  $\tup{t_{i,a}} \in \alg{M}^n$ such that $\tup{t_{i,a}}(j) = 0_\alg{M}$ for $j
  \in [n], j \neq i$ and $\tup{t_{i,a}}(i) = a$. It follows from
  $\tup{z} \in P_M(f')$ that we have $\tup{z}(\tup{t_{i,a}}) =
  \tup{z}(i, a) \leq f'(\tup{t_{i,a}}) \leq 0$ for any $a \in A$ and
  $i \in [n]$. \qed

  \flb{Claim C. Any $\tup{z} \in B_M(f') \subseteq P_M(f')$ satisfies
  $\tup{z}(\tup{1}_{\alg{M}^n}) = f'(\tup{1}_{\alg{M}^n})$.}

  Follows from the definition of $B_M(f')$ \qed

  Finally, $f'(\tup{1}_{\alg{M}^n}) = \min_{\tup{x} \in \alg{M}^n} f(\tup{x})$
  which follows from the definition of $f'$. From
  Lemma~\ref{diam:lem:greedy-Bf} and Lemma~\ref{diam:lem:spmod->inPf} it now
  follows that LHS $\leq$ RHS holds. To prove the existence of a
  integer valued vector, note that the vector from
  Lemma~\ref{diam:lem:greedy-Bf} is integer valued if $f'$ is integer
  valued and $f'$ is integer valued if $f$ is integer valued. \qed
\end{proof}

We can reformulate Theorem~\ref{diam:th:minmax} to relate the minimum
of a submodular function $f$ to the maximum of a certain function
defined over the polyhedron $\{ \tup{x} \in P_M(f) \mid \tup{x} \leq
0 \}$. To do this we define a function $S : \mathbb{R}^{[n] \times
A} \rightarrow \mathbb{R}$ as follows
\[
S(\tup{x}) = \sum_{i=1}^n \min_{a \in A} \tup{x}(i, a) + \max_{a \in A} \tup{x}(i, a) .
\]
We then get the following corollary.
\begin{corollary}
\[
\min_{\tup{y} \in \alg{M}^n} f(\tup{y}) = 
\max \left\{ S(\tup{z}) \ \big| \ \tup{z} \in P_M(f), \tup{z} \leq 0 \right\} .
\]
\end{corollary}
\begin{proof}
Follows from Theorem~\ref{diam:th:minmax} by two observations. If
$\tup{z}$ is unified, then $\tup{z}(\tup{1_{\alg{M}^n}}) =
S(\tup{z})$. Furthermore, any vector $\tup{z}$ can be turned into a
unified vector $\tup{z'}$ such that $\tup{z'} \leq \tup{z}$ and
$S(\tup{z}) = \tup{z'}(\tup{1_{\alg{M}^n}})$. (To construct $\tup{z'}$
from $\tup{z}$, for each $i \in [n]$, choose some
$p_i \in \maxarg_{a \in A} \tup{z}(i, a)$ and let $\tup{z'}(i, p_i)
= \tup{z}(i, p_i)$ and for $a \in A, a \neq p_i$ let $\tup{z'}(i, a)
= \min_{a \in A} \tup{z}(i, a)$.) \qed
\end{proof}

One might ask if there is any reason to believe that the min--max
characterisation given by Theorem~\ref{diam:th:minmax} is the ``right'' way
to look at this problem.  That is, can this min--max relation give
insight into the complexity of minimising submodular functions over
diamonds? Theorem~\ref{diam:th:minmax} is used in Section~\ref{diam:sec:good} to
get a good characterisation of submodular function minimisation over
diamonds, so it certainly gets us somewhere. In Section~\ref{diam:sec:find-min}
we present a pseudo-polynomial time algorithm which uses $P_M(f)$, but
it does not use Theorem~\ref{diam:th:minmax}. Additionally,
Theorem~\ref{diam:th:minmax} is in some sense fairly similar
to~\eqref{diam:eq:set-minmax}. In particular, in both cases the vectors are
functions from the atoms of the lattices to the real numbers and when
a vector is applied to a tuple (or a subset) it is computed as a sum
over the coordinates of the vector and the tuple. Furthermore, in this
sum the bottom of the lattice ($0_\alg{M}$ in the diamond case and
$\emptyset$ in the set case) do not contribute to the sum. There are
of course differences as well. The most obvious one is, perhaps, that
there is no element in the set case analogous to $1_\alg{M}$ in the diamond
case. Considering that, as far as we know, all combinatorial
algorithms for submodular set function minimisation is based
on~\eqref{diam:eq:set-minmax} and the similarity between
Theorem~\ref{diam:th:minmax} and~\eqref{diam:eq:set-minmax} one could hope that
Theorem~\ref{diam:th:minmax} could be the basis for a polynomial time
combinatorial algorithm for SFM$(\alg{M})$.

\ignore{
\begin{figure}
  \begin{flushleft}
    \textbf{Input:} Vector $\tup{x} \in \mathbb{R}^{[n] \times A}$ \\
    \textbf{Output:} Vector $\tup{x'} \in \mathbb{R}^{[n] \times A}$ \\
    \textbf{Unify$(\tup{x})$:}
  \end{flushleft}
  \begin{enumerate}
    \item For $i \in [n]$:
    \item $\quad$ Let $p_i = \maxarg_{a \in A} \tup{x}(i, a)$
    \item $\quad$ Set $\tup{x'}(i, p_i) \leftarrow \tup{x}(i, p_i)$
    \item $\quad$ For $a \in A$ such that $a \neq p_i$:
    \item $\quad \quad$ Set $\tup{x'}(i, a) \leftarrow  \min_{a \in A} \tup{x}(i, a)$
    \item Return $\tup{x'}$
  \end{enumerate}
  \caption{The procedure Unify.} \label{diam:fig:unify}
\end{figure}
}

The following theorem is an analog to the second equality in
Edmonds' min--max theorem for submodular set
functions~\eqref{diam:eq:set-minmax}.
\begin{theorem} \label{diam:th:Pf=Bf}
  Let $f : \alg{M}^n \rightarrow \mathbb{R}$ be a submodular function such
  that $f(\tup{0}_{\alg{M}^n}) = 0$, then
  \begin{align}
  \min_{\tup{x} \in \alg{M}^n} f(\tup{x}) &=
  \max \left\{ \tup{z}(\tup{1_{\alg{M}^n}}) \ \big| \ \tup{z} \in P_M(f), \tup{z} \leq 0, \text{ $\tup{z}$ is unified} \right\} \notag \\
  &= \max \left\{ \tup{x}^-(\tup{1_{\alg{M}^n}}) \ \big| \ \tup{x} \in B_M(f), \text{ $\tup{x}^-$ is unified} \right\} . \notag
  \end{align}
\end{theorem}
\begin{proof}
  We prove that
  \begin{align}
  &\max \left\{ \tup{z}(\tup{1_{\alg{M}^n}}) \ \big| \ 
                \tup{z} \in P_M(f), \tup{z} \leq 0, \text{$\tup{z}$ is unified} \right\} = \notag \\
  &\max \left\{ \tup{x}^-(\tup{1_{\alg{M}^n}}) \ \big| \ 
                \tup{x} \in B_M(f), \text{$\tup{x}$ is unified} \right\} . \notag
  \end{align}
  The result then follows from Theorem~\ref{diam:th:minmax}.

  Let $\tup{x}$ be a vector which maximises the right hand side. It is
  clear that $\tup{x}^- \in P_M(f)$, $\tup{x}^- \leq 0$, and that
  $\tup{x}^-$ is unified. It follows that LHS $\geq$ RHS.

  Conversely, let $\tup{z}$ be a vector which maximises the left hand
  side. We will define a sequence of vectors
  $\tup{x_0}, \tup{x_1}, \ldots$. We start with $\tup{x_0} = \tup{z}$
  and for $j \geq 0$ we define $\tup{x_{j+1}}$ from $\tup{x_j}$
  according to the construction below.
  \begin{enumerate}  
  \item If there is some $i \in [n]$ and $p \in \maxarg_{a \in
    A} \tup{x}(i, a)$ such that $\alpha' > 0$ where
    \[
    \alpha' = \max \{ \alpha \in \mathbb{R} \mid \tup{x} + \alpha \tup{\chi_{i, p}} \in P_M(f) \} ,
    \]
    then let $\tup{x_{j+1}} = \tup{x_j} + \alpha' \cdot \tup{\chi_{i, p}}$.

  \item Otherwise, if there is some $i \in [n]$ and
    $p \in \maxarg_{a \in A} \tup{x}(i, a)$ such that $\alpha' > 0$ where
    \[
    \alpha' = \max \{ \alpha \in \mathbb{R} \mid 
                        \tup{x_j} + \alpha \cdot (\tup{\chi_i} - \tup{\chi_{i, p}}) \in P_M(f) \} ,
    \]
    then let $a$ be some atom distinct from $p$, let $m
  = \min \{ \alpha', \tup{x}(i, p) - \tup{x}(i, a)\}$, and let
  $\tup{x_{j+1}} = \tup{x_j} + m \cdot (\tup{\chi_i} - \tup{\chi_{i,
  p}})$.
  \end{enumerate}

  We make four observations of this construction.
  \begin{itemize}
    \item If we reach the second step, then $\maxarg_{a \in A} \tup{x}(i, a)$ is a one element set.
    \item For every $j$ the vector $\tup{x_j}$ is unified.
    \item For every $j$, $\tup{x_{j+1}} \geq \tup{x_j}$.
    \item For every $j$, $\tup{x_j} \in P(f)$.
  \end{itemize}
  These observations all follows directly from the construction above.
  It is not hard to convince oneself that there is an integer $m$ such
  that $\tup{x_m} = \tup{x_{m+1}}$ (and thus all vectors constructed
  after $m$ are equal). To see this, note that for a fixed $i \in [n]$
  if some atom $a$ is increased in step 1, then this atom will not be
  increased again at coordinate $i$. Let $\tup{y}$ denote the vector
  $\tup{x_m}$.

  Note that
  $\tup{x_{j+1}}^-(\tup{1}_{\alg{M}^n}) \geq \tup{x_{j}}^-(\tup{1}_{\alg{M}^n})$
  for all $j$. Hence in particular
  $\tup{y}^-(\tup{1}_{\alg{M}^n}) \geq \tup{z}(\tup{1}_{\alg{M}^n})$.
  As we have already proved that LHS $\geq$ RHS it now remains to prove
  that $\tup{y} \in B_M(f)$. As we already know that $\tup{y} \in
  P_M(f)$ this reduces to proving $\tup{y}(\tup{1_{\alg{M}^n}}) =
  f(\tup{1}_{\alg{M}^n})$.

  Let $\tup{p}$ be a tuple such that for $i \in [n]$ we have
  $\tup{p}(i) = \max_{a \in A} \tup{y}(i, a)$. As
  $\tup{y} = \tup{x_m} = \tup{x_{m+1}}$, it follows that for each $k \in [n]$ there
  is an atom $a \in A, a \neq \tup{p}(k)$ and tuples $\tup{t_k},
  \tup{t'_k} \in \alg{M}^n, \tup{p}(k) \lattleq \tup{t_k}(k), a \lattleq
  \tup{t'_k}(k)$ such that $\tup{t_k}$ and $\tup{t'_k}$ are
  $\tup{y}$-tight.  Now let,
  \[
  \tup{t} = \biglub_{k \in [n]} \tup{t_k} \lub \tup{t'_k} .
  \]
  As $\tup{y} \in P_M(f)$ it follows from Lemma~\ref{diam:lem:tight} that
  $\tup{y}(\tup{t}) = f(\tup{t})$. Note that for each $k \in [n]$ we
  have $(\tup{t_k} \lub \tup{t'_k})(k) = 1_\alg{M}$, it follows that
  $\tup{t} = \tup{1}_{\alg{M}^n}$ and hence $\tup{y} \in B_M(f)$. We
  conclude that LHS $\leq$ RHS. \qed
\end{proof}

\section{A Good Characterisation} \label{diam:sec:good}
In this section we show that there are membership proofs for $P_M(f)$
which can be checked in time polynomial in $n$. By using
Theorem~\ref{diam:th:minmax} this will lead to the existence of proofs that
can be checked in time polynomial in $n$ of the fact that a certain
tuple minimises a submodular function.
The following lemma states that if $\tup{a}$ and $\tup{b}$ are
$\tup{x}$-tight, then so are $\tup{a} \glb \tup{b}$ and $\tup{a} \lub
\tup{b}$. This simple result will be used repeatedly in the subsequent
parts of the paper.
\begin{lemma} \label{diam:lem:tight}
  Let $f : \alg{M}^n \rightarrow \mathbb{R}$ be a submodular function. Let
  $\tup{x} \in P_M(f)$ be a vector and let $\tup{a}, \tup{b} \in \alg{M}^n$
  be $\tup{x}$-tight tuples. Then, $\tup{a} \lub \tup{b}$ and $\tup{a}
  \glb \tup{b}$ are $\tup{x}$-tight.
\end{lemma}
\begin{proof}
  \[
  \tup{x}(\tup{a} \lub \tup{b}) + \tup{x}(\tup{a} \glb \tup{b}) \leq
  f(\tup{a} \lub \tup{b}) + f(\tup{a} \glb \tup{b}) \leq
  f(\tup{a}) + f(\tup{b}) = \tup{x}(\tup{a}) + \tup{x}(\tup{b})
  \]
  The first inequality follows from $\tup{x} \in P_M(f)$, the second
  follows from the submodularity of $f$. The equality follows from the
  assumptions in the lemma. Note that $\tup{x}(\tup{a}) +
  \tup{x}(\tup{b}) \leq \tup{x}(\tup{a} \lub \tup{b}) +
  \tup{x}(\tup{a} \glb \tup{b})$.  Since $\tup{x}(\tup{a} \lub
  \tup{b}) \leq f(\tup{a} \lub \tup{b})$ and $\tup{x}(\tup{a} \glb
  \tup{b}) \leq f(\tup{a} \glb \tup{b})$, it follows that
  $\tup{x}(\tup{a} \lub \tup{b}) = f(\tup{a} \lub \tup{b})$ and
  $\tup{x}(\tup{a} \glb \tup{b}) = f(\tup{a} \glb \tup{b})$. \qed
\end{proof}

The following lemma is an important part of the main result in this
section.
\begin{lemma} \label{diam:lem:atmost-one-01}
  Let $\tup{c} \in \mathbb{R}^{[n] \times A}$ and assume that $\tup{x}$ maximises
  $\langle \tup{x}, \tup{c} \rangle$ over $P_M(f)$. Furthermore,
  assume that $\tup{a}, \tup{b} \in \alg{M}^n, \tup{a} \lattleq \tup{b}$ are
  $\tup{x}$-tight and for all $\tup{t} \in \alg{M}^n$ such that $\tup{a}
  \lattl \tup{t} \lattl \tup{b}$ the tuple $\tup{t}$ is not
  $\tup{x}$-tight. Then, there is at most one coordinate $i \in [n]$
  such that $\tup{a}(i) = 0_\alg{M}$ and $\tup{b}(i) = 1_\alg{M}$.
\end{lemma}
\begin{proof}
Assume that there is another coordinate $j \in [n], j \neq i$ such
that $\tup{a}(j) = 0_\alg{M}$ and $\tup{b}(j) = 1_\alg{M}$. We can assume, without
loss of generality, that
\[
\sum_{x \in A} \tup{c}(i, x) > \sum_{x \in A} \tup{c}(j, x) .
\]
\ignore{(FIXME: equality?)}  Let $\delta > 0$ and let $\tup{x'} =
\tup{x} + \delta \tup{\chi_i} - \delta \tup{\chi_j}$.  We cannot have
$\tup{x'} \in P_M(f)$ for any $\delta > 0$, because then $\tup{x}$ is
not optimal. As $\tup{x'} \not \in P_M(f)$ there is some
$\tup{x}$-tight tuple $\tup{t} \in \alg{M}^n$ such that ($\tup{t}(i) \in A$
and $\tup{t}(j) = 0_\alg{M}$) or ($\tup{t}(i) = 1_\alg{M}$ and $\tup{t}(j) \in
\{0_\alg{M}\} \cup A$). In either case, it follows from
Lemma~\ref{diam:lem:tight} that $\tup{t'} = (\tup{b} \glb \tup{t}) \lub
\tup{a}$ is $\tup{x}$-tight, which is a contradiction as $\tup{a}
\lattl \tup{t'} \lattl \tup{b}$. \qed
\end{proof}
The key lemma of this section is the following result. We will use
this lemma together with Lemma~\ref{diam:lem:atmost-one-01} in the proof of
the main result of this section (Theorem~\ref{diam:th:coNP}).

\begin{lemma} \label{diam:lem:poly-verify}
Let $n$ be a positive integer and let $f : \alg{M}^n \rightarrow \mathbb{R}$
be submodular which is provided to us by a value-giving oracle. Let
$\tup{x} \in \mathbb{R}^{[n] \times A}$ and $\tup{a}, \tup{b} \in \alg{M}^n$
such that $\tup{a} \lattleq \tup{b}$, $\tup{a}$ is $\tup{x}$-tight,
and there are at most $k$ coordinates $i \in [n]$ such that
$\tup{a}(i) = 0_\alg{M}$ and $\tup{b}(i) = 1_\alg{M}$. Under the assumption that
for all $\tup{t} \lattleq \tup{a}$ we have $\tup{x}(\tup{t}) \leq
f(\tup{t})$ it can be verified in time $O\left(n^{k+c} \right)$ that
$\tup{x}(\tup{y}) \leq f(\tup{y})$ holds for all $\tup{y} \lattleq
\tup{b}$, for some fixed constant $c$.
\end{lemma}
\begin{proof}
Let $I \subseteq [n]$ be the set of coordinates such that $i \in I$ if
and only if $\tup{a}(i) = 0_\alg{M}$ and $\tup{b}(i) = 1_\alg{M}$ and let $J = \{
j \in [n] \mid \tup{a}(j) \neq \tup{b}(j), j \not \in I \}$. Let $Z =
\{ \tup{z} \in \alg{M}^n \mid \forall i \not \in I: \tup{z}(i) = 0_\alg{M}
\}$. For a subset $Y = \{y_1, y_2, \ldots, y_m\}$ of $J$ and $\tup{z}
\in Z$ define $g_{\tup{z}} : 2^J \rightarrow \mathbb{R}$ as
\[
g_{\tup{z}}(Y) = f(\tup{a}[y_1 = \tup{b}(y_1), \ldots, y_m = \tup{b}(y_m)] \lub \tup{z}) .
\]
We claim that $g_{\tup{z}}$ is a submodular set function. Let $\tup{z}
\in Z$ and let $C = \{c_1, c_2, \ldots, c_k\}$ and $D = \{d_1, d_2,
\ldots, d_l\}$ be two arbitrary subsets of $J$. Define $\tup{c},
\tup{d} \in \alg{M}^n$ as $\tup{a}[c_1 = \tup{b}(c_1), \ldots, c_k =
\tup{b}(c_k)] \lub \tup{z}$ and $\tup{a}[d_1 = \tup{b}(d_1), \ldots, d_l =
\tup{b}(d_l)] \lub \tup{z}$, respectively. We now get
\[
g_{\tup{z}}(C) + g_{\tup{z}}(D) =    f(\tup{c}) + f(\tup{d})
            \geq f(\tup{c} \glb \tup{d}) + f(\tup{c} \lub \tup{d})
            =    g_{\tup{z}}(C \cap D) + g_{\tup{z}}(C \cup D) .
\]
Hence $g_{\tup{z}}$ is submodular for each $\tup{z} \in Z$. For a
subset $Y = \{y_1, y_2, \ldots, y_m\}$ of $J$ define $h_{\tup{z}} :
2^J \rightarrow \mathbb{R}$ as
\[
h_{\tup{z}}(Y) = \tup{x}(\tup{a}[y_1 = \tup{b}(y_1), \ldots, y_m = \tup{b}(y_m)] \lub \tup{z}) .
\]
We claim that $-h_{\tup{z}}$ is a submodular set function for each $\tup{z} \in
Z$. As above, let $C = \{c_1, c_2, \ldots, c_k\}$ and $D = \{d_1, d_2,
\ldots, d_l\}$ be two arbitrary subsets of $J$ and let $\tup{c} =
\tup{a}[c_1 = \tup{b}(c_1), \ldots, c_k = \tup{b}(c_k)] \lub \tup{z}$ and
$\tup{d} = \tup{a}[d_1 = \tup{b}(d_1), \ldots, d_l =
\tup{b}(d_l)] \lub \tup{z}$, then
\[
h_{\tup{z}}(C) + h_{\tup{z}}(D) =    \tup{x}(\tup{c}) + \tup{x}(\tup{d})
               \leq \tup{x}(\tup{c} \glb \tup{d}) + \tup{x}(\tup{c} \lub \tup{d})
                =    h_{\tup{z}}(C \cap D) + h_{\tup{z}}(C \cup D) .
\]
Hence, $-h_{\tup{z}}$ is submodular. Let $Y = \{y_1, y_2, \ldots, y_m\}$ be an
arbitrary subset of $J$ and let $\tup{z} \in Z$. For a fixed $k$ the inequalities
\begin{align}
&\tup{x}(\tup{a}[y_1 = \tup{b}(y_1), \ldots, y_m = \tup{b}(y_m)] \lub \tup{z}) \leq \notag \\
&f(\tup{a}[y_1 = \tup{b}(y_1), \ldots, y_m = \tup{b}(y_m)] \lub \tup{z}) \notag \\
&\iff \notag \\
&0 \leq g_{\tup{z}}(Y) - h_{\tup{z}}(Y) \label{diam:eq:submod-min}
\end{align}
can be verified to hold for every $Y \subseteq J$ and $\tup{z} \in Z$
in time polynomial in $n$ as, for each $x \in X$, the RHS
of~\eqref{diam:eq:submod-min} is a submodular set function in
$Y$. Conversely, if~\eqref{diam:eq:submod-min} does not hold for some $Y
\subseteq J$ and $x \in X$, then there is a tuple $\tup{t} \lattleq
\tup{b}$ such that $\tup{x}(\tup{t}) \not \leq f(\tup{t})$. To verify
that~\eqref{diam:eq:submod-min} holds for all $Y \subseteq J$ and $\tup{z}
\in Z$ find the minimum value of the RHS of~\eqref{diam:eq:submod-min} for
each $\tup{z} \in Z$ and compare it to 0 (note that $|Z|$ only depends
on $k$ and $|A|$). This can be done in time polynomial in $n$ by one
of the polynomial time algorithms for submodular function minimisation
(see, e.g.,~\cite{ellips-book,submod-min-P-iwata,submod-min-P-alex}
for descriptions of these algorithms).

Let $\tup{y} \in \alg{M}^n$ be a tuple such that $\tup{y} \lattleq
\tup{b}$. Note that if $\tup{a} \lattleq \tup{y}$, then it follows
from~\eqref{diam:eq:submod-min} that $\tup{x}(\tup{y}) \leq f(\tup{y})$.
For the sake of contradiction, assume that $\tup{x}(\tup{y}) \not \leq
f(\tup{y})$. By the submodularity of $f$ we get
\begin{align}
f(\tup{a} \lub \tup{y}) + f(\tup{a} \glb \tup{y}) \leq f(\tup{a}) + f(\tup{y}) . \label{diam:eq:submod-ay}
\end{align}
As $\tup{a} \lattleq \tup{a} \lub \tup{y} \lattleq \tup{b}$ it follows
from~\eqref{diam:eq:submod-min} that $\tup{x}(\tup{a} \lub \tup{y}) \leq
f(\tup{a} \lub \tup{y})$.  Furthermore, $\tup{y} \glb \tup{a} \lattleq
\tup{a}$ so by the assumptions in the lemma $\tup{x}(\tup{a} \glb
\tup{y}) \leq f(\tup{a} \glb \tup{y})$. By the choice of $\tup{a}$ and
$\tup{y}$ we get $\tup{x}(\tup{a}) = f(\tup{a})$ and $f(\tup{y}) <
\tup{x}(\tup{y})$.  It follows that
\begin{align}
\tup{x}(\tup{a} \lub \tup{y}) + \tup{x}(\tup{a} \glb \tup{y}) \leq f(\tup{a} \lub \tup{y}) + f(\tup{a} \glb \tup{y}) \label{diam:eq:y-f}
\end{align}
and
\begin{align}
f(\tup{a}) + f(\tup{y}) < \tup{x}(\tup{a}) + \tup{x}(\tup{y}) . \label{diam:eq:ff-yy}
\end{align}
But
\begin{align}
\tup{x}(\tup{a}) + \tup{x}(\tup{y}) \leq \tup{x}(\tup{a} \lub \tup{y}) + \tup{x}(\tup{a} \glb \tup{y}) \label{diam:eq:yy-yy}
\end{align}
so we get a contradiction by
combining~\eqref{diam:eq:y-f},~\eqref{diam:eq:submod-ay}, \eqref{diam:eq:ff-yy},
and~\eqref{diam:eq:yy-yy}. \qed
\end{proof}

Before we prove the main result of this section we need a few basic
facts about polyhedrons. Let $P \subseteq \mathbb{R}^n$ be a
polyhedron. The \emph{lineality space} of $P$, denoted by
$\text{lin.space}\ P$, is the set of vectors $\tup{x}$ such that there
is a vector $\tup{y} \in P$ and $\lambda \tup{x} + \tup{y} \in P$ for
all $\lambda \in \mathbb{R}$.  The \emph{characteristic cone} of $P$,
denoted by $\text{char.cone}\ P$, is the set of vectors $\tup{x} \in
\mathbb{R}^n$ such that for all $\tup{y} \in P$ and $\lambda \geq 0$
we have $\lambda \tup{x} + \tup{y} \in P$.

Given a submodular function $f$, it is not hard to see that the
characteristic cone of $P_M(f)$ are the vectors $\tup{x} \in
\mathbb{R}^{[n] \times A}$ such that $\tup{x} \leq 0$. Furthermore, the lineality
space of $P_M(f)$ is $\{ \tup{0} \}$. Given a polyhedron $P$ such that
$\text{lin.space}\ P = \{ \tup{0} \}$, it is well-known (see,
e.g,~\cite[Chapter~8]{linear-int-prog}) that any $\tup{x} \in P$ can
be represented as
$\tup{x} = \sum_{i = 1}^{n+1} \lambda_i \tup{y_i} + \tup{c}$
where $\tup{y_1}, \ldots, \tup{y_{n+1}}$ are vertices of $P$,
$\tup{c} \in \textnormal{char.cone}\ P$, $\sum_{i=1}^{n+1} \lambda_i =
1$, and $\lambda_i \geq 0$ for all $i$. (That is, $\tup{x}$ is the sum
of a convex combination of some of the vertices of $P$ and a vector in
the characteristic cone of $P$.) The fact that $n+1$ vertices suffice
is also well-known and is a corollary to Carath{\'e}odory's
Theorem~\cite{caratheodory} (see~\cite[Chapter~7.7]{linear-int-prog}
for a proof of the theorem and the corollary). We state this result
adapted to our setting as the following theorem.
\begin{theorem} \label{diam:th:caratheodory}
Let $f : \alg{M}^n \rightarrow \mathbb{R}$ be submodular and let $\tup{x}
\in P_M(f)$. Let $N = n \cdot |A|$. There are vertices $\tup{y_1},
\ldots, \tup{y_{N+1}}$ of $P_M(f)$, coefficients $\lambda_1, \ldots,
\lambda_{N+1} \in \mathbb{R}$, and a vector $\tup{c} \in
\mathbb{R}^{[n] \times A}$ such that
\[
\tup{x} = \sum_{i = 1}^{n+1} \lambda_i \tup{y_i} + \tup{c} ,
\]
$\tup{c} \leq 0$, $\sum_{i=1}^{N+1} \lambda_i = 1$, and $\lambda_i
\geq 0$ for each $i \in [N+1]$.
\end{theorem}

We start with showing that the vertices of $P_M(f)$ can be encoded in
not too many bits. This is needed in the proof of
Theorem~\ref{diam:th:coNP}.

Let $m$ be a positive integer. Given a set of vector $X \subseteq
\mathbb{R}^m$ we use $\text{conv}(X)$ to denote the convex hull of $X$
and $\text{cone}(X)$ to denote
\[
\{ \lambda_1 x_1 + \ldots + \lambda_t x_t \mid t \in \mathbb{N}, x_1, \ldots, x_t \in X, \lambda_1, \ldots, \lambda_t \geq 0 \} .
\]
\begin{definition}[Facet- and vertex-complexity~\cite{ellips-book}]
Let $P \subseteq \mathbb{R}^m$ be a polyhedron and let $\phi$ and
$\nu$ be positive integers.
\begin{itemize}
  \item $P$ has \emph{facet-complexity at most $\phi$} if there exists
  a system of linear inequalities with rational coefficients that has
  solution set $P$ and such that any equation can be encoded with at
  most $\phi$ bits. If $P = \mathbb{R}^m$ we require that $\phi \geq
  m+1$.

  \item $P$ has \emph{vertex-complexity at most $\nu$} if there exist
  finite sets $V$ and $E$ of rational vectors such that $P =
  \text{conv}(V) + \text{cone}(E)$ and such that each vector in $V$
  and $E$ can be encoded with at most $\nu$ bits. If $P = \emptyset$
  we require that $\nu \geq m$.
\end{itemize}
\end{definition}

\begin{lemma}[Part of Lemma~6.2.4 in~\cite{ellips-book}] \label{diam:lem:facet-vertex-compl}
  Let $P \subseteq \mathbb{R}^m$ be a polyhedron. If $P$ has
    facet-complexity at most $\phi$, then $P$ has vertex-complexity at
    most $4m^2 \phi$.
\end{lemma}

\begin{lemma} \label{diam:lem:vertex-compl-PD}
  There is a constant $c$ such that for any submodular $f : \alg{M}^n
  \rightarrow \mathbb{Z}$ the polyhedron $P_M(f)$ has
  vertex-complexity at most
  \[
   c \cdot |A| n^3 \cdot \log \max(|f|) .
  \]
\end{lemma}
\begin{proof}
  From the definition of $P_M(f)$ it follows that $P_M(f)$ has
  facet-complexity at most $c \cdot |A| n \cdot \log \max(|f|)$.  for
  some constant $c$. The lemma now follows from
  Lemma~\ref{diam:lem:facet-vertex-compl}. \qed
\end{proof}

Lemma~\ref{diam:lem:vertex-compl-PD} tells us that the vertices of $P_M(f)$
can be encoded with not too many bits (that is, the size is bounded by
a polynomial in $n$ and $\log \max(|f|)$). We are now ready to prove
the main theorem in this section, that $\alg{M}$ is well-characterised.
\begin{theorem} \label{diam:th:coNP}
  For every $k \geq 3$ the lattice $\alg{M}_k$ is well-characterised.
\end{theorem}
As usual we let $\alg{M}$ denote an arbitrary diamond. The idea in the proof
is that any point in $P_M(f)$ can be represented as a convex
combination of at most $n|A|+1$ vertices of $P_M(f)$ (this is
Carath{\'e}odory's theorem). Furthermore, by
Lemma~\ref{diam:lem:atmost-one-01} and an iterated use of
Lemma~\ref{diam:lem:poly-verify} there are membership proofs for the
vertices of $P_M(f)$ which can be checked in polynomial time. Hence,
we get membership proofs for all of $P_M(f)$ which can be checked
efficiently and by Theorem~\ref{diam:th:minmax} we obtain the result.

\begin{proof}
  Let $f : \alg{M}^n \rightarrow \mathbb{Z}$ be a submodular function and let
  $m$ be some integer. We will show that if $\min_{\tup{t} \in \alg{M}^n}
  f(\tup{t}) = m$, then there is a proof of this fact which can be
  checked in time polynomial in $n$.

  We can assume that $f(\tup{0}_{\alg{M}^n}) = 0$ as $\tup{t} \mapsto
  f(\tup{t}) - f(\tup{0}_{\alg{M}^n})$ is submodular.
  Let $N = n \cdot |A|$. The proof consists of a tuple $\tup{m} \in
  \alg{M}^n$, $N+1$ vectors $\tup{x_1}, \ldots, \tup{x_{N+1}} \in
  \mathbb{R}^{[n] \times A}$, for each $i \in [N+1]$ a sequence
  $\tup{t_i^1}, \ldots, \tup{t_i^{2n}} \in \alg{M}^n$ of tuples, and finally
  an integer-valued vector $\tup{c} \in \mathbb{R}^{[n] \times A}$.
  To verify the proof we first find $\tup{\lambda} = (\lambda_1,
  \ldots, \lambda_{N+1}) \in \mathbb{R}^{N+1}$ and $\tup{y} \in
  \mathbb{R}^{[n] \times A}$ such that
  \begin{align}
    \tup{y} \leq 0, \quad
    \sum_{i=1}^{N+1} \lambda_i \tup{x_i} + \tup{y} = \tup{c}, \quad
    \tup{\lambda} \geq 0, \quad \text{ and } \quad
    \sum_{i=1}^{N+1} \lambda_i = 1 . \label{diam:eq:syseq-verify}
  \end{align}
  This can be done in time polynomial in $n$. Reject the proof if
  there are no solutions to~\eqref{diam:eq:syseq-verify}. We proceed by
  checking that for each $i$
  \begin{itemize}
    \item $\tup{0}_{\alg{M}^n} = \tup{t_i^1} \lattleq \tup{t_i^2} \lattleq
    \ldots \lattleq \tup{t_i^{2n-1}} \lattleq \tup{t_i^{2n}} =
    \tup{1}_{\alg{M}^n}$, and

    \item $\tup{t_i^1}, \ldots, \tup{t_i^{2n}}$ are $\tup{x_i}$-tight, and

    \item for any $j \in [2n-1]$ there it at most one coordinate $l \in
      [n]$ such that $\tup{t_i^j}(l) = 0_\alg{M}$ and $\tup{t_i^{j+1}}(l) =
      1_\alg{M}$.
  \end{itemize}
  Reject the proof if any of these checks fail. We now want to verify
  that $\tup{x_i} \in P_M(f)$, this can be done by using
  Lemma~\ref{diam:lem:poly-verify} repeatedly. For $j = 1, 2, \ldots, 2n-1$
  we use the algorithm in Lemma~\ref{diam:lem:poly-verify} with $\tup{a} =
  \tup{t_i^j}$ and $\tup{b} = \tup{t_i^{j+1}}$. If all invocations of
  the algorithm succeeds we can conclude that $\tup{x_i} \in P_M(f)$,
  otherwise the proof is rejected. Finally, compute
  \[
  \tup{c} = \sum_{i=1}^{N+1} \lambda_i \tup{x_i} + \tup{y}
  \]
  and accept the proof if $\tup{c} \leq 0$, $\tup{c}$ is unified, and
  $\tup{c}(\tup{1_{\alg{M}^n}}) = f(\tup{m}) = m$.

  \medskip
  We now prove that this proof system is sound and complete.

  \flb{Completeness} (That is, if $m = \min_{\tup{y} \in \alg{M}^n}
  f(\tup{y})$ then there is a proof which the verifier accept.)
  By Theorem~\ref{diam:th:minmax} there is a unified integer-valued vector
  $\tup{c}$ such that $\tup{c} \in P_M(f)$, $\tup{c} \leq 0$ and $m =
  \tup{c}(\tup{1_{\alg{M}^n}})$. By Theorem~\ref{diam:th:caratheodory} there are vectors
  $\tup{x_1}, \ldots, \tup{x_{N+1}}$ such that for each $i \in [N+1]$
  $\tup{x_i}$ is a vertex of $P_M(f)$ and $\tup{c}$ is the sum of a
  convex combination of $\tup{x_1}, \ldots, \tup{x_{N+1}}$ with
  coefficients $\tup{\lambda} = (\lambda_1, \ldots, \lambda_{N+1})$
  and some vector $\tup{y} \in \text{char.cone}\ P_M(f)$, hence
  $\tup{\lambda}, \tup{y}$ is a solution to~\eqref{diam:eq:syseq-verify}.


  As for each $i \in [N+1]$ the vector $\tup{x_i}$ is a vertex of
  $P_M(f)$ it follows from Lemma~\ref{diam:lem:atmost-one-01} (and the
  observation that $\tup{0_{\alg{M}^n}}$ and $\tup{1_{\alg{M}^n}}$ are
  $\tup{x_i}$-tight) that there is a sequence of tuples $\tup{t_i^1},
  \tup{t_i^2}, \ldots, \tup{t_i^{2n}}$ such that $\tup{0_{\alg{M}^n}} =
  \tup{t_i^1} \lattleq \tup{t_i^2} \lattleq \ldots \lattleq
  \tup{t_i^{2n}} = \tup{1_{\alg{M}^n}}$ and for each $j \in [2n-1]$ there is
  at most one $l \in [n]$ such that $\tup{t_i^j}(l) = 0_\alg{M}$ and
  $\tup{t_i^{j+1}}(l) = 1_\alg{M}$.
  It follows that this proof is accepted by the verifier.

  \flb{Soundness} (That is, if there is a proof which the verifier
  accepts, then $m = \min_{\tup{y} \in \alg{M}^n} f(\tup{y})$.)
  As the verifier accepted the proof it follows from
  Lemma~\ref{diam:lem:poly-verify} that $\tup{x_i} \in P_M(f)$ for each $i
  \in [N+1]$. As $\tup{\lambda}$ and $\tup{y}$ is a solution
  to~\eqref{diam:eq:syseq-verify} it follows that $\tup{c} \in P_M(f)$ (it
  is a sum of a convex combination of some vectors contained in
  $P_M(f)$ and a vector in $\text{char.cone}\ P_M(f)$).
  From the acceptance of the verifier it also follows that $\tup{c}
  \geq 0$, $\tup{c}$ is unified, and $m = f(\tup{m}) =
  \tup{c}(\tup{1_{\alg{M}^n}})$. It now follows from Theorem~\ref{diam:th:minmax}
  that $m = \min_{\tup{y} \in \alg{M}^n} f(\tup{y})$. \qed
\end{proof}

In the proof system above, instead of letting the verifier
solve~\eqref{diam:eq:syseq-verify} we could have required that
$\tup{\lambda}$ and $\tup{y}$ are given in the proof. However, it is
not obvious that $\tup{\lambda}$ and $\tup{y}$ can be encoded in
$O(n^{k+c})$ bits (for some constant $c$). This follows from the
approach taken above by the fact that there are polynomial-time
algorithms for finding solutions to systems of linear inequalities and Lemma~\ref{diam:lem:vertex-compl-PD}.

Note that the vectors given in the proof do not need to be vertices of
$P_M(f)$. However, by using the tight tuples and by repeatedly using
Lemma~\ref{diam:lem:poly-verify} we can verify that the given vectors are
in fact contained in $P_M(f)$ anyway. By Lemma~\ref{diam:lem:atmost-one-01}
vectors and tight tuples always exist which satisfies the conditions
above (namely, if we chose some appropriate vertices of $P_M(f)$).

The following lemma, which uses Lemma~\ref{diam:lem:poly-verify}
essentially as we use it in Theorem~\ref{diam:th:coNP}, will be useful to
us in Section~\ref{diam:sec:find-min}.

\begin{lemma} \label{diam:lem:opt-tight-chain}
  Let $k$ be some fixed positive integer. Let $f : \alg{M}^n \rightarrow
  \mathbb{Z}$ be submodular and let $\tup{x}$ be a vector in
  $P_M(f)$. Let $\tup{t_1}, \ldots, \tup{t_m} \in \alg{M}^n$ be
  $\tup{x}$-tight tuples such that $\tup{0_{\alg{M}^n}} = \tup{t_1} \lattl
  \ldots \lattl \tup{t_m} = \tup{1_{\alg{M}^n}}$ and for each $j \in [m-1]$
  there is at most $k$ distinct $i_1, i_2, \ldots, i_k \in [n]$ such
  that $\tup{t_j}(i_1) = \tup{t_j}(i_2) = \ldots = \tup{t_j}(i_k) =
  0_\alg{M}$ and $\tup{t_{j+1}}(i_1) = \tup{t_{j+1}}(i_2) = \ldots =
  \tup{t_{j+1}}(i_k) = 1_\alg{M}$.

  For $i \in [m]$ let $E_i \subseteq I(\tup{t_i})$ such that $\tup{e}
  \in E_i$ if and only if $\langle \tup{e}, \tup{x} \rangle =
  f(\tup{t_i})$. Given $\tup{c} \in \mathbb{Q}^{[n] \times A}$,
  $\tup{x}$, and $\tup{t_1}, \ldots, \tup{t_m}$ it is possible to
  compute $\max \langle \tup{c}, \tup{y} \rangle$ subject to $\tup{y}
  \in P_M(f)$ and $\langle \tup{e}, \tup{y} \rangle = f(\tup{t_i})$
  for all $i \in [m]$ and $\tup{e} \in E_i$ in time polynomial in $n$,
  $\log \max(|f|)$ and the encoding length of $\tup{c}$.
\end{lemma}
Note that we do not require that the running time depend polynomially
on $k$.

\begin{proof}
  We construct a separation algorithm for the polyhedron
  \begin{align}
  \left\{ \tup{y} \in P_M(f) \mid \forall i \in [m], \tup{e} \in E_i: \langle \tup{e}, \tup{y} \rangle = f(\tup{t_i}) \right\} . \label{diam:eq:opt-tight-chain-poly}
  \end{align}
  The lemma then follows from the equivalence of separation and
  optimisation given by the Ellipsoid algorithm.

  Given a vector $\tup{y} \in \mathbb{R}^{[n] \times A}$ we first test
  that for all $i \in [m]$ and $\tup{e} \in E_i$ we have $\langle
  \tup{e}, \tup{y} \rangle = f(\tup{t_i})$. For each $i \in [m]$ we do
  this as follows: For each $j \in [n]$ such that $\tup{t_i}(j) = 1_\alg{M}$
  the set of pairs of atoms $a, b \in A$ such that $\tup{x}(j, a) +
  \tup{x}(j, b)$ is maximised must be a subset of the set of pairs of
  atoms $a', b' \in A$ such that $\tup{y}(j, a') + \tup{y}(j, b')$ is
  maximised. (Otherwise there is some $\tup{e} \in E_i$ such that
  $\langle \tup{x}, \tup{e} \rangle = f(\tup{t_i}) \neq \langle
  \tup{y}, \tup{e} \rangle$.) If this is the case then $\langle
  \tup{y}, \tup{e} \rangle = f(\tup{t_i})$ for \emph{all} $\tup{e} \in
  E_i$ if and only if $\langle \tup{y}, \tup{e} \rangle =
  f(\tup{t_i})$ for \emph{some} $\tup{e} \in E_i$.

  Note that this test can be done in polynomial time in $n$ and $\log
  \max(|f|)$ as $m \leq |A| \cdot n$. We can then use the algorithm in
  Lemma~\ref{diam:lem:poly-verify} to test if $\tup{y} \in P_M(f)$. By
  combining these two tests we have a separation oracle
  for~\eqref{diam:eq:opt-tight-chain-poly} and hence the lemma follows. \qed
\end{proof}

\section{Finding the Minimum Value} \label{diam:sec:find-min}

In this section we will show that there is an algorithm which finds
the minimum value of a submodular $f : \alg{M}^n \rightarrow \mathbb{Z}$ in
time polynomial in $n$ and $\max(|f|)$. Note that from an algorithm
which computes $\min_{\tup{t} \in \alg{M}^n} f(\tup{t})$ one can construct
an algorithm to find a minimiser of $f$, i.e., find a tuple $\tup{y}
\in \alg{M}^n$ such that $f(\tup{y}) = \min_{\tup{t} \in \alg{M}^n}
f(\tup{t})$. This can be done by for each $x \in \alg{M}$ minimising $f_x :
\alg{M}^{n-1} \rightarrow \mathbb{R}$ defined by $f_x(\tup{t}) = f(x,
\tup{t})$. If $\min_{\tup{t} \in \alg{M}^{n-1}} f_x(\tup{t}) = \min_{\tup{t}
\in \alg{M}^n} f(\tup{t})$, then there is a minimiser $\tup{y} \in \alg{M}^n$ to
$f$ such that $\tup{y}(1) = x$. By iterating this procedure $n$ times
one finds a minimiser of $f$.

We start with a high level description of the algorithm. The starting
point is the separation problem for $P_M(f)$ and the observation that
$\tup{0} \in P_M(f)$ if and only if $\min_{\tup{t} \in \alg{M}^n} f(\tup{t})
\geq 0$. Hence, given an algorithm for deciding if $\tup{0}$ is
contained in $P_M(f)$ we can apply a binary search strategy to find a
minimiser of $f$. (Note that for any $c \in \mathbb{R}$ the function
$f + c$ is submodular if $f$ is submodular.)

In each iteration $i$ of the algorithm we maintain an upper bound
$u_i$ and lower bound $l_i$ on $\min_{\tup{t} \in \alg{M}^n} f(\tup{t})$. If
$\tup{0} \in P_M(f - (u_i-l_i)/2)$ (note that $f - (u_i - l_i)/2$,
i.e., the function $f' : \alg{M}^n \rightarrow \mathbb{R}$ defined by
$f'(\tup{t}) = f(\tup{t}) - (u_i-l_i)/2$, is submodular if $f$ is
submodular), we iterate the algorithm with $u_{i+1} = u_i$ and
$l_{i+1} = (u_i - l_i)/2$. Otherwise, if $\tup{0} \not \in P_M(f)$, we
set $u_{i+1} = (u_i - l_i)/2$ and $l_{i+1} = l_i$. For an initial upper
bound we can use $u_1 = f(\tup{0}_{\alg{M}^n})$. To find a lower bound $l_1$
we can use Theorem~\ref{diam:th:Pf=Bf} together with the greedy algorithm
in Lemma~\ref{diam:lem:greedy-Bf}. The running time of this algorithm is
$O(S \cdot \log \max(|f|) + n)$, where $S$ is the time taken to decide
if $\tup{0} \in P_M(f)$.

By the equivalence of separation and optimisation given by the
Ellipsoid algorithm it is sufficient to solve the optimisation problem
for $P_M(f)$. (The results we will need which are related to the
Ellipsoid algorithm are given in Subsection~\ref{diam:sec:ellips}. We refer
the reader to~\cite{ellips-book} for an in-depth treatment of the
theory related to this topic.) In the optimisation problem we are
given $\tup{c} \in \mathbb{Q}^{[n] \times A}$ and are supposed to
solve $\max \langle \tup{c}, \tup{y} \rangle, \tup{y} \in P_M(f)$. To
get the running time we are aiming for we must do this in time
polynomial in $n$, $\max(|f|)$ and the encoding length of $\tup{c}$.

To solve this problem our algorithm starts with a vertex of $P_M(f)$
and either finds an adjacent vertex with a strictly better measure or
concludes that no such vertex exists. (This technique is called the
primal-dual method, see~\cite[Section~12.1]{linear-int-prog}.) The
initial vertex is found by the greedy algorithm in
Lemma~\ref{diam:lem:greedy-Bf}. To make this approach run in
pseudo-polynomial time two parts are needed. The first one is that the
existence of a pseudo-polynomial algorithm to go from vertex to a
better one or conclude that no such vertex exists. We present an
algorithm for this in Section~\ref{diam:sec:improve}. The other part is
that we must ensure that the algorithm makes enough progress in each
iteration so that we get the bound on the running time we are aiming
for. To this end, we prove in Section~\ref{diam:sec:half} that the vertices
of $P_M(f)$ are half-integral.

This section is organised as follows: in Subsection~\ref{diam:sec:ellips}
state some results we will need related to the Ellipsoid algorithm. In
Subsection~\ref{diam:sec:struct} we prove a couple of results of the
structure of the vertices of $P_M(f)$. We also show that a submodular
function can be turned into a strictly submodular function such that
any minimiser of the latter is also a minimiser of the former. This
will be useful to us in the subsequent parts of the algorithm. In
Subsection~\ref{diam:sec:half} we prove that the vertices of $P_M(f)$ are
half-integral. Finally, in Subsection~\ref{diam:sec:improve} we show how we
can go from one vertex of $P_M(f)$ to a better one (if there is one)
and how this can be used to construct an optimisation algorithm for
$P_M(f)$. \ignore{The oracle-pseudo-tractability for the diamonds is given in
Corollary~\ref{diam:cor:o-p-t}.}

\subsection{The Ellipsoid Algorithm} \label{diam:sec:ellips}

In this subsection we present some definitions and results which are
related to the Ellipsoid algorithm. They are all from the
book~\cite{ellips-book}. As in~\cite{ellips-book} we make the
general assumption that for any oracle $O$ there is an integer $c$
such that when $O$ is given input data of length $n$ the length of the
output is $O(n^c)$.

\begin{definition}[Oracle-polynomial time]
An algorithm $A$, with access to an oracle $O$, runs in
\emph{oracle-polynomial time} if there is an integer $c$ such that given
any input of length $n$ $A$ makes $O(n^c)$ calls to $O$ and performs
$O(n^c)$ additional primitive operations.
\end{definition}

This is definition~6.2.2c in~\cite{ellips-book}.
\begin{definition}[Well-described polyhedron]
A \emph{well-described polyhedron} is a triple $(P; n, \phi)$ where $P
\subseteq \mathbb{R}^n$ is a polyhedron with facet-complexity at most
$\phi$. The encoding length of $(P; n, \phi)$ is $\phi + n$.
\end{definition}

This is definition~6.2.1 in~\cite{ellips-book}.
\begin{definition}[Strong optimization problem]
Given a polyhedron $P \subseteq \mathbb{R}^n$ and a vector $\tup{c}
\in \mathbb{Q}^n$, either
\begin{itemize}
  \item assert that $P$ is empty, or
  \item find a vector $\tup{y} \in P$ maximising $\langle \tup{c}, \tup{x} \rangle$ over $P$, or
  \item find a vector $\tup{z} \in \text{char.cone}\ P$ such that $\langle \tup{c},
  \tup{z} \rangle \geq 1$.
\end{itemize}
\end{definition}

This is definition~2.1.4 from~\cite{ellips-book}.
\begin{definition}[Strong separation problem]
  Given a vector $\tup{y} \in \mathbb{R}^n$, decide whether $\tup{y}
  \in P$, and if not, find a hyperplane that separates $\tup{y}$ from
  $P$; more exactly, find a vector $\tup{c} \in \mathbb{R}^n$ such
  that $\langle \tup{c}, \tup{y} \rangle > \max \{ \langle \tup{c},
  \tup{x} \rangle \mid \tup{x} \in P \}$.
\end{definition}

This is a part of Theorem~6.4.9 in~\cite{ellips-book}.
\begin{theorem}
Let $(P; n, \phi)$ be a well-described polyhedron. The strong
separation problem and strong optimisation problem for $(P; n, \phi)$
can be solved in oracle-polynomial time given an oracle for the other
problem.
\end{theorem}

\subsection{The Structure of the Vertices of $P_M(f)$} \label{diam:sec:struct}

The following lemma is stated in~\cite[Theorem~2.1]{combopt-schrijver}
for the boolean lattice. Essentially the same proof works for modular
lattices. We give a version of the lemma specialised to $\alg{M}^n$.
\begin{lemma} \label{diam:lem:schrijver-ineq}
Let $\tup{t}, \tup{u} \in \alg{M}^n$ such that $\tup{t} \not \lattleq
\tup{u}$ and $\tup{u} \not \lattleq \tup{t}$, then
\[
\rho(\tup{t})(2n - \rho(\tup{t})) + \rho(\tup{u})(2n - \rho(\tup{u})) >
\rho(\tup{t} \glb \tup{u}) (2n - \rho(\tup{t} \glb \tup{u})) + \rho(\tup{t} \lub \tup{u})(2n - \rho(\tup{t} \lub \tup{u}))
\]
\end{lemma}
\begin{proof}
Let $\alpha = \rho(\tup{t} \glb \tup{u}), \beta = \rho(\tup{t}) -
\rho(\tup{t} \glb \tup{u})$, $\gamma = \rho(\tup{u}) - \rho(\tup{t}
\glb \tup{u})$, and $\delta = 2n - \rho(\tup{t} \lub \tup{u})$. Then
the LHS is equal to
\[
(\alpha + \beta)(\gamma + \delta) + (\alpha + \gamma)(\beta + \delta) = 2\alpha \delta + 2 \beta \gamma + \alpha \gamma + \beta \delta + \alpha \beta + \gamma \delta
\]
as $\rho$ is modular (i.e., $\rho(\tup{t}) + \rho(\tup{u}) =
\rho(\tup{t} \lub \tup{u}) + \rho(\tup{t} \glb \tup{u})$). The RHS is
equal to
\[
\alpha (\beta + \gamma + \delta) + (\alpha + \beta + \gamma) \delta = 2 \alpha \delta + \alpha \gamma + \beta \delta + \alpha \beta + \gamma \delta .
\]
Since $\beta \gamma > 0$ the lemma follows. \qed
\end{proof}

The lemma above tells us that the function $\tup{t} \mapsto
\rho(\tup{t})(2n - \rho(\tup{t}))$ is strictly submodular. Note that
if $f : \alg{M}^n \rightarrow \mathbb{R}$ is submodular, then $f$ can be
turned into a strictly submodular function $f' : \alg{M}^n \rightarrow
\mathbb{R}$ by $f'(\tup{t}) = f(\tup{t}) + \epsilon \rho(\tup{t}) (2n
- \rho(\tup{t}))$. Observe that if $\epsilon > 0$ is chosen small
enough then any minimiser of $f'$ is also a minimiser of $f$. Strictly
submodular functions are an interesting subset of the submodular
functions due to this observation and the following lemma.
\begin{lemma} \label{diam:lem:strict-chain}
Let $f : \alg{M}^n \rightarrow \mathbb{R}$ be strictly submodular and let
$\tup{x}$ be a vertex of $P_M(f)$. Then, the $\tup{x}$-tight tuples
form a chain.
\end{lemma}
\begin{proof}
Assume, for the sake of contradiction, that $\tup{t}, \tup{u} \in \alg{M}^n$
are $\tup{x}$-tight and $\tup{t} \not \lattleq \tup{u}$ and $\tup{u}
\not \lattleq \tup{t}$. It follows that
\begin{align}
\tup{x}(\tup{t}) + \tup{x}(\tup{u}) = f(\tup{t}) + f(\tup{u}) > f(\tup{t} \glb \tup{u}) + f(\tup{t} \lub \tup{u}) = \tup{x}(\tup{u} \lub \tup{v}) + \tup{x}(\tup{u} \glb \tup{v}) . \label{diam:eq:strict-submod}
\end{align}
The last equality follows from the fact that the $\tup{x}$-tight
tuples are closed under $\glb$ and $\lub$. However,
\eqref{diam:eq:strict-submod} contradicts the supermodularity of
$\tup{x}$. \qed
\end{proof}

Lemma~\ref{diam:lem:strict-chain} tells us that, for strictly submodular
$f$, for each vertex $\tup{x}$ of $P_M(f)$ the set of $\tup{x}$-tight
tuples is a chain in $\alg{M}^n$. As the dimension of $P_M(f)$ is $|A|n$,
for every vertex $\tup{x}$ of $P_M(f)$ there are $|A|n$ linearly
independent inequalities which are satisfied with equality by
$\tup{x}$. This means that for every such $\tup{x}$ there is a chain
$\tup{t_1} \lattl \tup{t_2} \lattl \ldots \lattl \tup{t_m}$ in $\alg{M}^n$
and linearly independent vectors $\tup{e_1}, \tup{e_2}, \ldots,
\tup{e_{|A|n}}$ such that for each $i \in [|A|n]$ there is some $j(i)
\in [m]$ such that $\tup{e_i} \in I(\tup{t_{j(i)}})$ and for all $i
\in [|A|n]$ we have $\langle \tup{e_i}, \tup{x} \rangle =
f(\tup{t_{j(i)}})$. Furthermore, $\tup{x}$ is the only vector which
satisfies $\langle \tup{e_i}, \tup{x} \rangle = f(\tup{t_i})$ for all
$i \in [|A|n]$.

For general (not necessarily strict) submodular functions the set of
$\tup{x}$-tight tuples is not necessarily a chain, but one can prove
that for every vertex there is a chain of tuples such that some subset
of the inequalities induced by the tight tuples characterises the
vertex. That is, given the subset of inequalities induced by such a
chain of tight tuples there is only one point in $P_M(f)$ which
satisfies all the inequalities with equality. Formally we state this
as the following lemma.

\begin{lemma} \label{diam:lem:gen-chain}
Let $f : \alg{M}^n \rightarrow \mathbb{R}$ be submodular and let $\tup{x}$
be a vertex of $P_M(f)$. Then there is a chain $\tup{t_1} \lattl
\tup{t_2} \lattl \ldots \lattl \tup{t_m}$ in $\alg{M}^n$ and linearly
independent vectors $\tup{e_1}, \tup{e_2}, \ldots, \tup{e_{|A|n}}$,
which are all $\tup{x}$-tight, such that for each $i \in [|A|n]$ there
is some $j(i) \in [m]$ such that $\tup{e_i} \in I(\tup{t_{j(i)}})$.
\end{lemma}
\begin{proof}
Define $f'(\tup{t}) = f(\tup{t}) + \epsilon \rho(\tup{t})(2n -
\rho(\tup{t}))$ and choose $\epsilon > 0$ small. From
Lemma~\ref{diam:lem:schrijver-ineq} it follows that $f'$ is strictly
submodular.

Let $\tup{c} \in \mathbb{R}^{[n] \times A}$ such that $\tup{x}$ is the
unique optimum to $\max \langle \tup{c}, \tup{y} \rangle, \tup{y} \in
P_M(f)$. Let $\tup{x'}$ be a vertex of $P_M(f')$ which is an optimum
to $\max \langle \tup{c}, \tup{y} \rangle, \tup{y} \in P_M(f')$. From
Lemma~\ref{diam:lem:strict-chain} it follows that as $\tup{x'}$ is a vertex
of $P_M(f')$ there are $\tup{t_1} \lattl \tup{t_2} \lattl \ldots
\lattl \tup{t_m}$ and $\tup{x'}$-tight $\tup{e_1}, \ldots,
\tup{e_{n|A|}}$ as in the statement of the lemma. Let $\tup{e_1},
\ldots, \tup{e_{n|A|}}$ be the rows of the matrix $A$ and define
$\tup{b} = (f(\tup{t_{j(1)}}), \ldots, f(\tup{t_{j(m)}}))^T$ and
\[
\tup{\epsilon} = \epsilon \cdot (\rho(\tup{t_{j(1)}})(2n - \rho(\tup{t_{j(1)}})), \ldots,
                                 \rho(\tup{t_{j(m)}})(2n - \rho(\tup{t_{j(m)}})))^T .
\]
It follows that $\tup{x'} = A^{-1} (\tup{b} + \tup{\epsilon})$. We
proceed by establishing two claims.

\flb{Claim A. $A^{-1} \tup{b} \in P_M(f)$.}

To see this assume for the
sake of contradiction that $A^{-1} \tup{b} \not \in P_M(f)$, then
there is some $\tup{t} \in \alg{M}^n$ and $\tup{e} \in I(\tup{t})$ such that
$\tup{e} A^{-1} \tup{b} > f(\tup{t})$. However,
\begin{align}
\tup{e} \tup{x'} =
\tup{e} A^{-1} (\tup{b} + \tup{\epsilon}) =
\tup{e} A^{-1} \tup{b} + \tup{e} A^{-1} \tup{\epsilon} \leq
f(\tup{t}) + \epsilon \rho(\tup{t})(2n - \rho(\tup{t})) . \label{diam:eq:ineq-xp}
\end{align}
As $\tup{e} A^{-1} \tup{b} > f(\tup{t})$ we can choose some $\delta >
0$ such that $\tup{e} A^{-1} \tup{b} > f(\tup{t}) + \delta$. By
choosing $\epsilon$ so that $|\epsilon \rho(\tup{t})(2n -
\rho(\tup{t})) - \tup{e} A^{-1} \tup{\epsilon}| < \delta$ we get
\[
\tup{e} A^{-1} \tup{b} + \tup{e} A^{-1} \tup{\epsilon} > f(\tup{t}) + \epsilon \rho(\tup{t})(2n - \rho(\tup{t}))
\]
which contradicts~\eqref{diam:eq:ineq-xp}. We conclude that $A^{-1} \tup{b}
\in P_M(f)$. \qed

\flb{Claim B. $\langle A^{-1} \tup{b}, \tup{c} \rangle = \langle
\tup{x}, \tup{c} \rangle$.}

Assume, for the sake of contradiction,
that $\langle A^{-1} \tup{b}, \tup{c} \rangle < \langle \tup{x},
\tup{c} \rangle$. (The inequality $\leq$ follows from our choice of
$\tup{x}$ and Claim~A.) Note that $\tup{x} \in P_M(f) \subseteq
P_M(f')$. For sufficiently small $\epsilon$ we get
\[
\langle \tup{x'}, \tup{c} \rangle =
\langle A^{-1}(\tup{b} + \tup{\epsilon}), \tup{c} \rangle <
\langle \tup{x}, \tup{c} \rangle
\]
which contradicts the optimality of $\tup{x'}$. \qed

We have shown that for every vertex $\tup{x} \in P_M(f)$ there is some
vertex $\tup{x'} \in P_M(f')$ which satisfies some inequalities, given
by the matrix $A$, with equality. As $f'$ is strictly submodular it
follows from Lemma~\ref{diam:lem:strict-chain} that the $\tup{x'}$-tight
tuples form a chain. By Claim~A the inequalities in $A$ also defines a
point in $P_M(f)$. Furthermore, by Claim~B this point maximises
$\langle \tup{y}, \tup{c} \rangle$ over $P_M(f)$. By our choice of
$\tup{c}$ it follows that $A^{-1} \tup{b} = \tup{x}$. The lemma
follows. \qed
\end{proof}

\subsection{$P_M(f)$ is Half-integral} \label{diam:sec:half}
In this subsection we will prove that if $f : \alg{M}^n \rightarrow
\mathbb{Z}$ is submodular, then the vertices of $P_M(f)$ are
half-integral.

\begin{lemma} \label{diam:lem:vertex-unif}
Let $f : \alg{M}^n \rightarrow \mathbb{R}$ be submodular and let $\tup{x}$
be a vertex of $P_M(f)$. For each $i \in [n]$ there are three
possibilities
\begin{enumerate}
\item $\tup{x}(i, a) = \tup{x}(i, b)$ for all $a, b \in A$; or
\item there is exactly one atom $a' \in A$ such that $\tup{x}(i, a') > \min_{a \in A} \tup{x}(i, a)$; or
\item there is exactly one atom $a' \in A$ such that $\tup{x}(i, a') < \max_{a \in A} \tup{x}(i, a)$.
\end{enumerate}
\end{lemma}
\begin{proof}
  As $\tup{x}$ is a vertex of $P_M(f)$ there is a $\tup{c} \in
  \mathbb{R}^{[n] \times A}$ such that $\tup{x}$ is the unique optimum
  to $\langle \tup{c}, \tup{y} \rangle, \tup{y} \in P_M(f)$. As the
  optimum exist it follows that $\tup{c} \geq \tup{0}$. Assume, for
  the sake of contradiction, that there is a coordinate $i \in [n]$
  such that the statement of the lemma does not hold for $i$.  Let
  $A_1$ be the atoms $a' \in A$ which satisfies $\tup{x}(i, a') =
  \max_{a \in A} \tup{x}(i, a)$. Similarly, let $A_2$ be the atoms $a'
  \in A_2$ which satisfies $\tup{x}(i, a') = \max_{a \in A \setminus
  A_1} \tup{x}(i, a)$. Finally, let $A_3 = A \setminus (A_1 \cup
  A_2)$. We will first prove the following claim.

  \flb{Claim. There are distinct atoms $b, c \in A$ and
  $\tup{x}$-tight tuples $\tup{t_1}, \tup{t_2} \in \alg{M}^n$ such that
  $\tup{t_1}(i) = b$ and $\tup{t_2}(i) = c$, furthermore $b \in A_3$
  or $b, c \in A_2$.}

  If $a \in A_3$ let $\tup{x'} = \tup{x} + \delta \tup{\chi}(i, a)$
  for some small $\delta > 0$. As $\tup{x}$ is the unique optimum it
  follows that $\tup{x'} \not \in P_M(f)$ and hence there is an
  $\tup{x}$-tight tuple $\tup{t} \in \alg{M}^n$ such that $\tup{t}(i) =
  a$. So if $|A_3| \geq 2$, then the claim holds. Similarly, if $|A_1|
  \geq 2$, then any for any $a \in A \setminus A_1$ we get an
  $\tup{x}$-tight tuple $\tup{t}$ such that $\tup{t}(i) = a$. (Again
  this follows from considering the vector $\tup{x'} = \tup{x} +
  \delta \tup{\chi}(i, a)$.)

  So $|A_3| \leq 1$ and ($|A_1| = 1$ or $|A_1| \geq |A|-1$). If $|A_3| =
  0$ and ($|A_1| = 1$ or $|A_1| \geq |A|-1$), then the statement of the
  lemma holds, so we must have $|A_3| = 1$. This implies that $|A_1| =
  1$.

  Let $A_1 = \{a\}$. If $\tup{c}(i, a) \geq \sum_{b \in A_2} \tup{c}(i,
  b)$, then let $\tup{x'} = \tup{x} + \delta \tup{\chi}(i, a) - \delta
  \sum_{b \in A_2} \tup{\chi}(i, b)$ for some small $\delta > 0$. It
  follows that $\tup{x'} \not \in P_M(f)$ and hence there is an
  $\tup{x}$-tight tuple $\tup{t}$ with $\tup{t}(i) = a$. In the other
  case, when $\tup{c}(i, a) < \sum_{b \in A_2} \tup{c}(i, b)$, we let
  $\tup{x'} = \tup{x} - \delta \tup{\chi}(i, a) + \delta \sum_{b \in
  A_2} \tup{\chi}(i, b)$. It follows that there is an $\tup{x}$-tight
  tuple $\tup{t}$ with $\tup{t}(i) \in A_2$. \qed

  Let $b$ and $c$ be the atoms in the claim above and let $\tup{t_1}$
  and $\tup{t_2}$ be the $\tup{x}$-tight tuples in the claim.
%
%
  As $f$ is submodular we have
  \[
  f(\tup{t_1} \lub \tup{t_2}) + f(\tup{t_1} \glb \tup{t_2}) \leq f(\tup{t_1}) + f(\tup{t_2}) = \tup{x}(\tup{t_1}) + \tup{x}(\tup{t_2}) .
  \]
  From this inequality and the fact that $\tup{x} \in P_M(f)$ it
  follows that
  \begin{align}
  &\tup{x}(\tup{t_1} \lub \tup{t_2}) + \tup{x}(\tup{t_1} \glb \tup{t_2}) &\leq \notag \\
  &f(\tup{t_1} \lub \tup{t_2}) + f(\tup{t_1} \glb \tup{t_2}) &\leq \notag \\
  &\tup{x}(\tup{t_1}) + \tup{x}(\tup{t_2}) &\leq \notag \\
  &\tup{x}(\tup{t_1} \lub \tup{t_2}) + \tup{x}(\tup{t_1} \glb \tup{t_2}) . \notag
  \end{align}
  We conclude that $\tup{x}(\tup{t_1} \lub \tup{t_2}) +
  \tup{x}(\tup{t_1} \glb \tup{t_2}) = \tup{x}(\tup{t_1}) +
  \tup{x}(\tup{t_2})$. However, this leads to a contradiction:
  \begin{align}
  &\tup{x}(\tup{t_1}) + \tup{x}(\tup{t_2}) = \tup{x}(i, b) + \tup{x}(i, c) + \tup{x}(\tup{t_1}[i = 0_\alg{M}]) + \tup{x}(\tup{t_2}[i = 0_\alg{M}]) &\leq \notag \\
  &\tup{x}(i, b) + \tup{x}(i, c) + \tup{x}((\tup{t_1} \lub \tup{t_2})[i = 0_\alg{M}]) + \tup{x}(\tup{t_1} \glb \tup{t_2}) &< \notag \\
  &\tup{x}(\tup{t_1} \lub \tup{t_2}) + \tup{x}(\tup{t_1} \glb \tup{t_2}) \notag
  \end{align}
  So the coordinate $i$ cannot exist. \qed
\end{proof}

The lemma above can be strengthened if $|A| = 3$, in this case only 1
and 2 are possible. To see this, assume that $A = \{a_1, a_2, a_3\}$
and $\tup{x}(i, a_1) = \tup{x}(i, a_2) > \tup{x}(i, a_3)$. Let
$\tup{x'} = \tup{x} + \delta \tup{\chi}(i, a_1) - \delta \tup{\chi}(i,
a_2)$ (or $\tup{x'} = \tup{x} - \delta \tup{\chi}(i, a_1) + \delta
\tup{\chi}(i, a_2)$ if $\tup{c}(i, a_1) < \tup{c}(i, a_2)$). As
$\tup{x'} \not \in P_M(f)$ it follows that there is some
$\tup{x}$-tight tuple $\tup{t}$ with $\tup{t}(i) = a_1$ (or
$\tup{t}(i) = a_2$). We can then proceed as in the proof above.

We will need the following lemma
from~\cite{tight-bounds-2-approx-ilp2} in our proof of the
half-integrality of $P_M(f)$.
\begin{lemma} \label{diam:lem:col-sum-2}
Let $A$ be a $m \times n$ integral matrix satisfying
\[
\sum_{i = 1}^m |A_{ij}| \leq 2
\]
for $j \in [n]$. Then, for every square non-singular submatrix $S$ of
$A$, $S^{-1}$ is half-integral.
\end{lemma}

By combining Lemma~\ref{diam:lem:strict-chain}, Lemma~\ref{diam:lem:vertex-unif}
 and Lemma~\ref{diam:lem:col-sum-2} we are able the obtain the following
 theorem which asserts the half-integrality of $P_M(f)$.
\begin{theorem} \label{diam:th:half}
Let $f : \alg{M}^n \rightarrow \mathbb{Z}$ be submodular. For any vertex
$\tup{x}$ of $P_M(f)$ and any $i \in [n]$ and $a \in A$ we have $\tup{x}(i,
a) \in \{ 1/2 \cdot k \mid k \in \mathbb{Z} \}$.
\end{theorem}
\begin{proof}
Let $\tup{x}$ be a vertex of $P_M(f)$. By Lemma~\ref{diam:lem:gen-chain}
there is a chain of $\tup{x}$-tight tuples $\tup{t_1} \lattl \ldots
\lattl \tup{t_m}$ and linearly independent vectors $\tup{e_1}, \ldots,
\tup{e_{|A|n}}$ such that for each $i \in [|A|n]$ there is some $j(i)
\in [m]$ such that $\tup{e_i} \in I(\tup{t_{j(i)}})$. We can also
assume that for $i \leq i'$ we have $j(i) \leq j(i')$. Let $E$ be the
matrix with rows $\tup{e_1}, \ldots, \tup{e_{|A|n}}$, then $\tup{x}$
is the unique solution to $E \tup{x} = \tup{b}$, where
\[
\tup{b} = \left( f\left(\tup{t_{j(1)}}\right), f\left(\tup{t_{j(2)}}\right), \ldots, f\left(\tup{t_{j(|A|n)}}\right) \right)^T .
\]
Let $A = \{a_1, a_2, \ldots, a_{|A|}\}$. By
Lemma~\ref{diam:lem:vertex-unif} we can assume, without loss of generality,
that and $\tup{x}(i, a_2) = \tup{x}(i, a_3) = \ldots = \tup{x}(i,
a_{|A|})$.  If $\tup{x}(i, a_1) = \tup{x}(i, a_2) = \ldots =
\tup{x}(i, a_{|A|})$ we can identify $\tup{x}(i, a_1), \ldots,
\tup{x}(i, a_{|A|})$ without changing the set of solutions to $E
\tup{x} = \tup{b}$, in the other case when $\tup{x}(i, a_1) > \min_{a
\in A} \tup{x}(i, a)$ or $\tup{x}(i, a_1) < \max_{a \in A} \tup{x}(i,
a)$ we can identify $\tup{x}(i, a_2), \ldots, \tup{x}(i, a_{|A|})$
without changing the set of solutions to $E \tup{x} = \tup{b}$. After
having identified these variables we get a system of linear equations,
$E' \tup{x'} = \tup{b'}$, which has a unique solution. Furthermore,
the solution to $E' \tup{x'} = \tup{b'}$ is half-integral if and only
if $E \tup{x} = \tup{b}$ has a half-integral solution (that is, if and
only if $\tup{x}$ is half-integral). Let $X \subseteq [n] \times \{1,
2\}$ such that for each $i \in [n]$, $(i, 1) \in X$ and $(i, 2) \in X$
if and only if $\tup{x}(i, a_1) > \min_{a \in A} \tup{x}(i, a)$ or
$\tup{x}(i, a_1) < \max_{a \in A} \tup{x}(i, a)$. We can describe the
rows, $\tup{e'_1}, \tup{e'_2}, \ldots, \tup{e'_{|A|n}} \in
\mathbb{R}^X$ of $E'$ as follows
\begin{itemize}
\item if $\tup{x}(i, a_1) = \tup{x}(i, a_2) = \ldots = \tup{x}(i,
a_{|A|})$, then $\tup{e'_j}(i, 1) = \sum_{a \in A} \tup{e_j}(i, a)$;

\item otherwise (if $\tup{x}(i, a_1) > \min_{a \in A} \tup{x}(i, a)$
or $\tup{x}(i, a_1) < \max_{a \in A} \tup{x}(i, a)$), then
$\tup{e'_j}(i, 1) = \tup{e_j}(i, a_1)$ and $\tup{e'_j}(i, 2) = \sum_{a
\in A, a \neq a_1} \tup{e'_j}(a)$.
\end{itemize}
As the solution to $E' \tup{x'} = \tup{b'}$ and $E \tup{x} = \tup{b}$
are equal, modulo the identification of some of the variables, there
is a subset $R = \{r_1, r_2, \ldots, r_{|X|}\} \subseteq [|A|n]$ with
$r_1 < r_2 < \ldots < r_{|X|}$ such that the matrix $E''$,
with rows $\{ \tup{e'_i} \mid i \in R \}$, has an
inverse. Furthermore, this inverse is half-integral (that is,
$E''^{-1}$ is half-integral) if and only if the solution to $E'
\tup{x'} = \tup{b'}$ is half-integral.

It is easy to see that the entries of $E''$ are contained in
$\{0,1,2\}$. Furthermore, if $\tup{c}$ is an arbitrary column of
$E''$, then it is of the form $(0, \ldots, 0, 1, \ldots, 1, 2, \ldots,
2)^T$ or $(0, \ldots, 0, 1, \ldots, 1, 0, \ldots, 0)^T$ (in these
patterns, $x, \ldots, x$ means that $x$ occurs zero or more times). It
follows that for each $(i, k) \in X$ we have
\begin{align}
\sum_{l = 1}^{|X|} |\tup{e'_{r_{l+1}}}(i, k) - \tup{e'_{r_l}}(i, k)| \leq 2 . \label{diam:eq:leq-two}
\end{align}
Following the proof of Theorem~1 in~\cite{fract-matroid-match} we now
define
\[
U =
\left(
\begin{array}{ccccc}
1      & 0      & \cdots &  \cdots & 0     \\
-1     & \ddots & \ddots &         & \vdots \\
0      & \ddots & \ddots &  \ddots & \vdots \\
\vdots & \ddots & \ddots &  \ddots & 0     \\
0      & \cdots & 0      & -1      & 1
\end{array}
\right) .
\]
We can then express the inverse of $E''$ as $(U E'')^{-1}
U$. By~\eqref{diam:eq:leq-two} and Lemma~\ref{diam:lem:col-sum-2} it follows
that $(U E'')^{-1}$ is half-integral and hence $E''^{-1}$ is
half-integral as well, which implies that $\tup{x}$ is
half-integral. \qed
\end{proof}

\subsection{Finding Augmentations} \label{diam:sec:improve}
Let $f : \alg{M}^n \rightarrow \mathbb{Z}$ be submodular.  In this section
we will show that there is an algorithm which decides if $\tup{0} \in
P_M(f)$ in time polynomial in $n$ and $\max(|f|)$. The strategy of the
algorithm is to use the equivalence between separation and
optimisation given by the Ellipsoid algorithm and solve the
optimisation problem for $P_M(f)$ instead. In the optimisation problem
we are given $\tup{c} \in \mathbb{Q}^{[n] \times A}$ and are supposed
to find $\max \langle \tup{c}, \tup{y} \rangle, \tup{y} \in
P_M(f)$. This problem is solved by iterating an augmentation step in
which we are in some vertex $\tup{x}$ of $P_M(f)$ and wish to find
some vertex $\tup{x'}$, adjacent to $\tup{x}$, such that $\langle
\tup{c}, \tup{x'} \rangle > \langle \tup{c}, \tup{x} \rangle$.

Let $\tup{c} \in \mathbb{Q}^{[n] \times A}$ and assume that we want to
solve $\max \langle \tup{c}, \tup{y} \rangle, \tup{y} \in P_M(f)$. Let
$T$ be the set of all $\tup{x}$-tight tuples and let $E \subseteq
\cup_{\tup{t} \in T} I(\tup{t})$ such that $\tup{e} \in E$ if and only
if there is some $\tup{t} \in T$ with $\tup{e} \in I(\tup{t})$ and
$\langle \tup{e}, \tup{x} \rangle = f(\tup{t})$. Finding a vector
$\tup{y} \in \mathbb{R}^{[n] \times A}$ such that there is some
$\delta > 0$ which satisfies $\langle \tup{c}, \tup{x} \rangle <
\langle \tup{c}, \tup{x} + \delta \tup{y} \rangle$ and $\tup{x} +
\delta \tup{y} \in P_M(f)$ or conclude that no such vector $\tup{y}$
exists is equivalent to solving the linear program
\ignore{
\begin{align}
\max \langle \tup{c}, \tup{z} \rangle \text{ subject to} \label{diam:eq:increasing} \\
\langle \tup{e}, \tup{z} \rangle \leq 0 \text{ for all } \tup{e} \in E \notag \\
\langle \tup{c}, \tup{z} \rangle \leq 1 . \notag
\end{align}
}
\begin{align} \label{diam:eq:increasing}
\max \langle \tup{c}, \tup{z} \rangle \text{ subject to }
\forall \tup{e} \in E: \langle \tup{e}, \tup{z} \rangle \leq 0 \text{ and }
\langle \tup{c}, \tup{z} \rangle \leq 1 .
\end{align}
(Here $\tup{z}$ contains the variables.) The optimum of this linear
program is $0$ if $\tup{x}$ is optimal and $1$ otherwise. The
separation problem for this polyhedron reduces to computing
\begin{align}
\max_{\tup{e} \in E} \langle \tup{e}, \tup{z} \rangle . \notag 
\end{align}
Define $f' : \alg{M}^n \rightarrow \mathbb{Z}$ as
$f'(\tup{t}) = (n^2 + 1) \cdot f(\tup{t}) + \rho(\tup{t})(2n - \rho(\tup{t}))$.
It is not hard to see that a minimiser of $f'$ is also a minimiser of
$f$. Furthermore, by Lemma~\ref{diam:lem:schrijver-ineq}, $f'$ is strictly
submodular. When minimising submodular functions we can thus assume
that the function is strictly submodular. By
Lemma~\ref{diam:lem:strict-chain} if $\tup{x}$ is a vertex of $P_M(f')$,
then $T$ (the $\tup{x}$-tight tuples) is a chain. This implies that
$|T| \leq 2n$.

\begin{lemma} \label{diam:lem:find-inc-dir}
If $f : \alg{M}^n \rightarrow \mathbb{Z}$ is strictly submodular and
$\tup{x}$ a vertex of $P_M(f)$, then the linear
program~\eqref{diam:eq:increasing} can be solved in time polynomial in $n$,
$\log \max(|f|)$ and the encoding length of $\tup{c}$. (Assuming that $T$ is
available to the algorithm.)
\end{lemma}
\begin{proof}
  As $f$ is strictly submodular it follows from
  Lemma~\ref{diam:lem:strict-chain} that $|T| \leq 2n$. Hence, the
  separation problem for~\eqref{diam:eq:increasing} can be solved in
  polynomial time. By the equivalence of separation and optimisation
  given by the Ellipsoid algorithm it follows
  that~\eqref{diam:eq:increasing} can be solved in time polynomial in $n$,
  $\log \max(|f|)$ and the encoding length of $\tup{c}$. (Note that even
  though $|T| \leq 2n$, the number of inequalities in $E$ may be
  exponential in $n$. In particular the tuple $\tup{1}_{\alg{M}^n}$ can
  induce as many as ${|A| \choose 2}^n$ inequalities.) \qed
\end{proof}
By the algorithm in Lemma~\ref{diam:lem:find-inc-dir} we can find an
optimal solution $\tup{z}$ to~\eqref{diam:eq:increasing}. We can use this
algorithm to find adjacent vertices which are better (if there are
any). We also need to find the largest $\delta > 0$ such that $\tup{x}
+ \delta \tup{z} \in P_M(f)$. We construct an algorithm for this in
Lemma~\ref{diam:lem:find-better}, but first we need a lemma.
\begin{lemma} \label{diam:lem:no-large-gap}
Let $\tup{y}$ be an optimal solution to~\eqref{diam:eq:increasing} which is
a vertex such that $\langle \tup{c}, \tup{y} \rangle = 1$. Assume that
there are $\tup{t_1}, \tup{t_2} \in \alg{M}^n$, $\tup{t_1} \lattl
\tup{t_2}$, and $\tup{e_1} \in E \cap I(\tup{t_1}), \tup{e_2} \in E
\cap I(\tup{t_2})$ such that $\langle \tup{e_1}, \tup{y} \rangle =
f(\tup{t_1})$ and $\langle \tup{e_2}, \tup{y} \rangle =
f(\tup{t_2})$. Furthermore, assume that there is no $\tup{u} \in T$
such that $\tup{t_1} \lattl \tup{u} \lattl \tup{t_2}$ with any
$\tup{e} \in E \cap I(\tup{u})$ and $\langle \tup{e}, \tup{y} \rangle
= f(\tup{u})$. Then, there are no three distinct coordinates $i, j, k
\in [n]$ such that $\tup{t_1}(i) = \tup{t_1}(j) = \tup{t_1}(k) = 0_\alg{M}$
and $\tup{t_2}(i) = \tup{t_2}(j) = \tup{t_2}(k) = 1_\alg{M}$.
\end{lemma}
\begin{proof}
Let $E' \subseteq E$ be the vectors which define tight inequalities
for $\tup{y}$. As $\tup{y}$ is a vertex and $\langle \tup{c}, \tup{y}
\rangle = 1$, it follows that the polyhedron $P = \{ \tup{z} \in
\mathbb{R}^{[n] \times A} \mid \langle \tup{e}, \tup{z} \rangle =
f(\tup{t}), \tup{e} \in E', \tup{e} \in I(\tup{t})\}$ is one
dimensional.

Let $\alpha, \beta \in \mathbb{R}$ be arbitrary and define $\tup{y'}
\in \mathbb{R}^{[n] \times A}$ by
\[
\tup{y'} = \tup{y} + (\alpha + \beta) \tup{\chi_i} - \alpha \tup{\chi_j} - \beta \tup{\chi_k} .
\]
From the non-existence of any $\tup{u} \in \alg{M}^n$ such that $\tup{t_1}
\lattl \tup{u} \lattl \tup{t_2}$ and $\tup{e} \in E \cap I(\tup{u})$,
$\langle \tup{e}, \tup{y} \rangle = f(\tup{u})$ it follows that
$\langle \tup{e}, \tup{y'} \rangle = f(\tup{t})$ for all $\tup{e} \in
E', \tup{e} \in I(\tup{t})$. However, this means that $\tup{y'} \in P$
and as $\alpha$ and $\beta$ where arbitrary it follows that $P$ is not
one-dimensional. This is a contradiction and the lemma follows.  \qed
\end{proof}

The following lemma is a crucial part of our pseudo-polynomial time
algorithm for SFM$(\alg{M})$. With the algorithm in this lemma we are able
to go from one vertex in $P_M(f)$ to a better one (if there is a
better one).

\begin{lemma} \label{diam:lem:find-better}
Let $f : \alg{M}^n \rightarrow \mathbb{Z}$ be a strictly submodular
function. Given $\tup{c} \in \mathbb{Q}^{[n] \times A}$, a vertex
$\tup{x}$ of $P_M(f)$, and the set of $\tup{x}$-tight tuples $T$,
there is an algorithm which is polynomial in $n$, $\log \max(|f|)$ and
the encoding length of $\tup{c}$ which finds a vertex $\tup{y} \in
P_M(f)$ such that $\langle \tup{c}, \tup{y} \rangle > \langle \tup{c},
\tup{x} \rangle$ or concludes that no such vertex exist. If $\tup{y}$
exists the set of $\tup{y}$-tight tuples can be computed within the
same time bound.
\end{lemma}
\begin{proof}
If there is such a vertex $\tup{y}$, then the value of the optimum of
the linear program~\eqref{diam:eq:increasing} is $1$. By
Lemma~\ref{diam:lem:find-inc-dir} this optimum $\tup{y'}$ can be found in
polynomial time. The set of tuples $T' \subseteq T$ which are
$\tup{y'}$-tight can be found in polynomial time (as $|T| \leq
2n$). Furthermore, by Lemma~\ref{diam:lem:no-large-gap} the gap between two
successive tuples in $T'$ is not too large. It follows from
Lemma~\ref{diam:lem:opt-tight-chain} that we can find a vertex $\tup{y}$ of
$P_M(f)$ such that $\langle \tup{c}, \tup{y} \rangle > \langle
\tup{c}, \tup{x} \rangle$ in polynomial time. 

It remains to find the rest of the $\tup{y'}$-tight tuples within the
stated time bound. By Lemma~\ref{diam:lem:no-large-gap} for any
consecutive tuples $\tup{a}, \tup{b}$ in $T'$ there are at most two
distinct coordinates $i,j \in [n]$ such that $\tup{a}(i) = \tup{a}(j)
= 0_\alg{M}$ and $\tup{b}(i) = \tup{b}(j) = 1_\alg{M}$. We will show
that for every such pair $\tup{a}, \tup{b}$ in $T'$ we can find the
$\tup{y}$-tight tuples $\tup{t}$ which satisfies $\tup{a} \lattl
\tup{t} \lattl \tup{b}$. To do this, for each $p, q \in M$, we find
the minimisers to the submodular function $f_{p,q}$ defined as
$f_{p,q}(\tup{x}) = f(\tup{x}[i = p, j = q]) -
\tup{y}(\tup{x}[i = p, j = q])$ over the set $X = \{\tup{x} \in M^n \mid \tup{a} \lattl \tup{x}
\lattl \tup{b}\}$. As $f$ is submodular and $\tup{y}$ is supermodular it follows that $f'$ is submodular.
To minimise $f_{p,q}$ over $X$ we can minimise it over at most $n^2
|A|^2$ intervals defined by $\{ \tup{x} \in M^n \mid \tup{a^*}
\lattleq \tup{x} \lattleq \tup{b_*} \}$ where $\tup{a} \prec
\tup{a^*}$ and $\tup{b_*} \prec \tup{b}$ (there are at most $n |A|$ choices for
$\tup{a^*}$ and at most $n |A|$ choices for $\tup{b_*}$).

Note that each of these intervals is a product of the two element
lattice and hence this minimisation can be done with the known
algorithms for minimising submodular set functions. We can use this
method to find all minimisers of $f_{p,q}$ in the interval we are
interested in. (When we have found one minimiser $\tup{m}$ we
iteratively minimise $f_{p,q}$ over the sets $\{ \tup{x} \in M^n \mid
\tup{a}
\lattl \tup{x} \lattl \tup{m} \}$ and $\{ \tup{x} \in M^n \mid \tup{m} \lattl
\tup{x} \lattl \tup{b} \}$.) As the $\tup{y}$-tight tuples is a chain in $M^n$
there are only a polynomial number of $\tup{y}$-tight tuples and hence
this step of the algorithm runs in polynomial time. Hence the set of
all $\tup{y}$-tight tuples can be found within the stated time
bound. \qed
\end{proof}

We are now finally ready to show the existence of a pseudo-polynomial
time separation algorithm for $P_M(f)$.

\begin{theorem} \label{diam:th:sep-0}
Let $f : \alg{M}^n \rightarrow \mathbb{Z}$ be submodular. It is possible to
decide if $\tup{0}$ is contained in $P_M(f)$ or not in time polynomial
in $n$ and $\max(|f|)$.
\end{theorem}
\begin{proof}
By the equivalence of separation and optimisation given by the
Ellipsoid algorithm there is an algorithm which decides if $\tup{0}$
is contained in $P_M(f)$ or not which makes use of an optimisation
oracle for $P_M(f)$. The number of calls to the optimisation oracle is
bounded by a polynomial in $n$ and $\log \max(|f|)$, furthermore the objective
function given to the optimisation oracle is given by a vector
$\tup{c} \in \mathbb{Q}^{[n] \times A}$ such that the encoding length
of $\tup{c}$ is bounded by a polynomial in $n$ and $\log \max(|f|)$.

To prove the lemma it is therefore sufficient to construct an
algorithm such that given $\tup{c} \in \mathbb{Z}^{[n] \times A}$
(there is no loss of generality in assuming that $\tup{c}$ is
integral, a simple scaling of $\tup{c}$ achieves this) it solves $\max
\langle \tup{y}, \tup{c} \rangle, \tup{y} \in P_M(f)$ in time
polynomial in $n$, $\max(|f|)$ and the size of the encoding of
$\tup{c}$. Let
$f'(\tup{t}) = (n^2+1) \cdot f(\tup{t}) + \rho(\tup{t})(2n - \rho(\tup{t}))$.
By Lemma~\ref{diam:lem:schrijver-ineq} $f'$ is strictly
submodular. Furthermore, it is easy to see that any minimiser of $f'$
is also a minimiser of $f$. By Lemma~\ref{diam:lem:strict-chain} each
vertex $\tup{x}$ of $P_M(f')$ is ``characterised'' of a chain of
$\tup{x}$-tight tuples.

The algorithm consists of a number of iterations. In iteration $j$ a
current vertex $\tup{x_j}$ of $P_M(f')$ is computed together with its
associated chain $\tup{C_j}$ of $\tup{x_j}$-tight tuples. The initial
vertex $\tup{x_0}$ and initial chain $\tup{C_0}$ is computed by the
greedy algorithm from Lemma~\ref{diam:lem:greedy-Bf}.

In iteration $j$, either $\tup{x_j}$ is the optimum or there is some
other vertex $\tup{x_{j+1}}$ such that $\langle \tup{x_{j+1}}, \tup{c}
\rangle > \langle \tup{x}, \tup{c} \rangle$. To find such an
$\tup{x_{j+1}}$ or conclude that no such vertex exists we use the
algorithm from Lemma~\ref{diam:lem:find-better}. In the case when
$\tup{x_{j+1}}$ exists we also get the chain $\tup{C_{j+1}}$ of
$\tup{x_{j+1}}$-tight from the algorithm in
Lemma~\ref{diam:lem:find-better}.

By Theorem~\ref{diam:th:half} the vertices of $P_M(f)$ are
half-integral. This implies that $\langle \tup{x_{j+1}}, \tup{c}
\rangle \geq \langle \tup{x_j}, \tup{c} \rangle + 1/2$. So the
algorithm is polynomial if we can prove that the optimum value is not
too far from the starting point $\tup{x}_0$. That is, the difference
between $\langle \tup{c}, \tup{x}_0 \rangle$ and $\max \langle
\tup{c}, \tup{y} \rangle, \tup{y} \in P_M(f)$ should be bounded by a
polynomial in $n$, $\max(|f|)$ and the encoding length of
$\tup{c}$. Note that as the size of the encoding of $\tup{c}$ is
bounded by a polynomial in $n$ and $\log \max(|f|)$ it follows that
$\max_{i \in [n], a \in A} |\tup{c}(i, a)|$ is bounded by a polynomial
in $n$ and $\max(|f|)$. Furthermore, as $\tup{x}_0$ is obtained by the
greedy algorithm it follows that for any $i \in [n], a \in A$ we have
$-2 \max(|f|) \leq \tup{x}_0(i, a)$. We now obtain the inequality
\begin{align}
&-2 \max(|f|) \cdot n|A| \left( \max_{i \in [n], a \in A} |\tup{c}(i, a)| \right) \leq
\langle \tup{c}, \tup{x}_0 \rangle \leq
\langle \tup{c}, \tup{y} \rangle \leq \notag \\
&\max(|f|) \cdot n|A| \left( \max_{i \in [n], a \in A} |\tup{c}(i, a)| \right) . \notag
\end{align}
From this inequality and the fact that $\max_{i \in [n], a \in A}
|\tup{c}(i, a)|$ is bounded by a polynomial in $n$ and $\max(|f|)$ it
follows that the difference between $\langle \tup{c}, \tup{x}_0
\rangle$ and $\langle \tup{c}, \tup{y} \rangle$ is bounded by a
polynomial in $n$ and $\max(|f|)$. As the objective function increases
by at least $1/2$ in each iteration this implies that the number of
iterations is bounded by a polynomial in $n$ and $\max(|f|)$. \qed
\end{proof}

From Theorem~\ref{diam:th:sep-0} we now get our desired result, a
pseudo-polynomial time algorithm for minimising submodular functions
over diamonds. The proof of this final step was given in
Section~\ref{diam:sec:find-min}.

\section{Conclusions and Open Problems} \label{diam:sec:concl}

The most obvious open problem is to find a polynomial time algorithm,
as opposed to a pseudo-polynomial time algorithm established in this
paper, for minimising submodular functions over diamonds. One possible
approach may be to use some kind of scaling technique see,
e.g.,~\cite{cap-scale-submodflow,submod-min-P-iwata}.
The pseudo-polynomial algorithm as it is presented here is very
inefficient: it consists of a nested application of the Ellipsoid
algorithm. Usually, one layer of the Ellipsoid algorithm is considered
to be too inefficient to be used in practise. It would clearly be
desirable to have a simpler and more efficient minimisation algorithm.

\bibliography{bibtexdb}
\end{document}